\documentclass{scrartcl}%
\usepackage{booktabs}
\usepackage{multirow}
\usepackage{skull}
\usepackage{amsmath}
\usepackage{amssymb}
\usepackage{amsfonts}
\usepackage{amsthm}
\usepackage{thmtools}		
\usepackage{mleftright}
\usepackage{thm-restate}
\usepackage[mathic=true]{mathtools}
\usepackage{fixmath} 

\newtheorem{theorem}{Theorem}%
\newtheorem{lemma}{Lemma}%
\newtheorem{invariant}{Invariant}%
\newtheorem{definition}{Definition}%

\usepackage{graphicx}
\usepackage{subcaption}

\usepackage{tikz}
\usetikzlibrary{arrows.meta}

\usepackage{enumitem}
\usepackage{xspace}

\usepackage{algorithm}
\usepackage{algorithmic}
\usepackage[algo2e,linesnumbered,ruled,vlined]{algorithm2e}
\DontPrintSemicolon
\SetKwComment{note}{$\triangleright$ }{}
\SetFuncSty{textsc}
\SetCommentSty{textit}
\SetDataSty{texttt}
\SetKwInput{Initial}{Initial}
\SetKwInput{Parameter}{Param.}
\SetKwInput{Input}{Input}
\SetKwInput{Data}{Data}
\SetKwInput{Output}{Output}
\SetKw{ForParallel}{parallel}
\SetKw{kwOr}{or}
\ResetInOut{Result} 

\SetCommentSty{mycommfont}

\newcommand{\tg}{\mathcal{G}\xspace}
\newcommand{\tge}{\mathcal{E}\xspace}

\usepackage[hidelinks]{hyperref}
\usepackage{cleveref}

\usepackage[authoryear,round]{natbib}
\begin{document}

\title{Inferring Tie Strength in Temporal Networks\thanks{A preliminary version of this paper was presented at the European Conference on Machine Learning and Principles and Practice of Knowledge Discovery in Databases 2022 ({ECMLPKDD} 2022)~\citep{oettershageninferring}.}
}

\author{{Lutz} {Oettershagen}\footnotetext{${}^\text{1}$This research is supported by the ERC Advanced Grant REBOUND (834862) and the EC H2020 RIA project SoBigData++ (871042).} \\KTH Royal Institute of Technology\\\texttt{lutzo@kth.se}\and 
{Athanasios L.} {Konstantinidis} \\Luiss University\\ \texttt{akonstantinidis@luiss.it}\and
{Giuseppe F.} {Italiano} \\Luiss University\\ \texttt{gitaliano@luiss.it}}

%
%
%
%
	\maketitle

\begin{abstract}
Inferring tie strengths in social networks is an essential task in social network analysis. Common approaches classify the ties as \emph{weak} and \emph{strong} ties based on the \emph{strong triadic closure (STC)}. The STC states that if for three nodes, $A$, $B$, and $C$, there are strong ties between $A$ and $B$, as well as $A$ and $C$, there has to be a (weak or strong) tie between $B$ and $C$. A variant of the STC called STC+ allows adding a few new weak edges to obtain improved solutions. 
So far, most works discuss the STC or STC+ in static networks. However, modern large-scale social networks are usually highly dynamic, providing user contacts and communications as streams of edge updates.
\emph{Temporal networks} capture these dynamics.
To apply the STC to temporal networks, we first generalize the STC and introduce a weighted version such that  empirical a priori knowledge given in the form of edge weights is respected by the STC. Similarly, we introduce a generalized weighted version of the STC+.
The weighted STC is hard to compute, and our main contribution is an efficient 2-approximation (resp.~3-approximation) 
streaming algorithm for the weighted STC (resp.~STC+) in temporal networks.
As a technical contribution, we introduce a fully dynamic $k$-approximation for the minimum weighted vertex cover problem in hypergraphs with edges of size $k$, which is a crucial component of our streaming algorithms.
An empirical evaluation shows that the weighted STC leads to solutions that better capture the a priori knowledge given by the edge weights than the non-weighted STC.
Moreover, we show that our streaming algorithm efficiently approximates the weighted STC in real-world large-scale social networks.

\end{abstract}

\section{Introduction}
Due to the explosive growth of online social networks and electronic communication, the automated inference of tie strengths is critical for many applications, e.g., advertisement, information dissemination, or understanding of complex human behavior~\citep{gilbert2008network,kahanda2009using}.
Users of large-scale social networks are commonly connected to hundreds or even thousands of other participants~\citep{kossinets2006empirical,mislove2007measurement}.
It is the typical case that these ties are not equally important. For example, in a social network, we can be connected with close friends as well as casual contacts.
Since a pioneering work of~\cite{granovetter1973strength}, the topic of tie strength inference has gained increasing attention fueled by the advent of online social networks and ubiquitous contact data.
Nowadays, ties strength inference in social networks is an extensively studied topic in the graph-mining community~\citep{GilbertK09,kahanda2009using,RozenshteinTG2017}.
A recent work by \cite{sintos2014using} introduced the \emph{strong triadic closure (STC)} property, where edges are classified as either \emph{strong} or \emph{weak}---for three persons with two \emph{strong} ties, there has to be a \emph{weak} or \emph{strong} third tie. Hence, if person $A$ is strongly connected to $B$, and $B$ is strongly connected to $C$, $A$ and $C$ are at least weakly connected.
The intuition is that if $A$ and $B$ are good friends, and $B$ and $C$ are good friends, $A$ and $C$ should at least know each other.

\begin{figure}[htb]\centering
\begin{subfigure}{0.49\linewidth}\centering
    \begin{tikzpicture}[scale=1]
        \node[style={}] (0,-1.5){};
        \begin{scope}[every node/.style={circle,thick,draw,minimum size=5mm,inner sep=0.5pt}]
            \node (a) at (0,2) {\emph{A}};
            \node (b) at (2,2) {\emph{B}};
            \node (c) at (3,0.5) {\emph{C}};
            \node (d) at (4,2) {\emph{D}};
        \end{scope}
        
        \begin{scope}
            \path [-] (a) edge[color=red,thick] node[above] {$10$} (b);
            \path [-] (b) edge[thick,dashed] node[above] {$1$} (d);
            \path [-] (b) edge[thick,dashed] node[left] {$1$} (c);
            \path [-] (c) edge[color=red,thick] node[right] {$2$} (d);
        \end{scope}
    \end{tikzpicture}
    \caption{Example for optimal weighted STC. The edge weights correspond to the number of communications.}
    \label{fig:examplea}
\end{subfigure}\hfill%
\begin{subfigure}{0.49\linewidth}\centering
    \begin{tikzpicture}[scale=1]
        \node[style={}] (0,-1.5){};
        \begin{scope}[every node/.style={circle,thick,draw,minimum size=5mm,inner sep=0.5pt}]
            \node (a) at (0,2) {\emph{A}};
            \node (b) at (2,2) {\emph{B}};
            \node (c) at (3,0.5) {\emph{C}};
            \node (d) at (4,2) {\emph{D}};
        \end{scope}
        
        \begin{scope}
            \path [-] (a) edge[thick,dashed] node[above,color=white] {0} (b);
            \path [-] (b) edge[color=red,thick] node[above] {} (d);
            \path [-] (b) edge[color=red,thick] node[left] {} (c);
            \path [-] (c) edge[color=red,thick] node[right] {} (d);
        \end{scope}
    \end{tikzpicture}
    \caption{Example for optimal non-weighted STC. Ignoring the weights leads to three strong edges.}
    \label{fig:exampleb}
\end{subfigure}
\caption{Example for the difference between weighted and non-weighted STC. Strong edges are highlighted in red, and weak edges are dashed.}
\label{fig:example}
\vspace*{-2mm}
\end{figure}
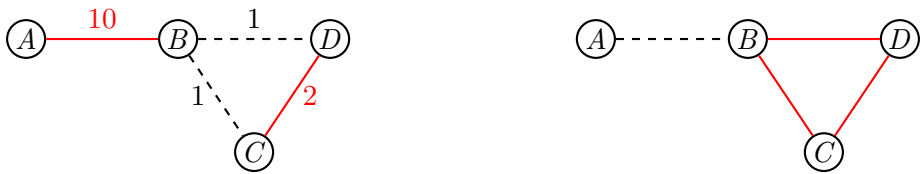

We first generalize the ideas of~\cite{sintos2014using} such that edge weights representing empirical tie strength are included in the computation of the STC. 
The idea is to consider edge weights that correspond to the empirical strength of the tie, e.g., the frequency or duration of communication between two persons. If this weight is high, we expect the tie to be strong, and we expect to be weak otherwise. However, we still want to fulfill the STC, and simple thresholding would not lead to correct results.
\Cref{fig:example} shows an example where we have a small social network consisting of four persons $A$, $B$, $C$, and $D$. In \Cref{fig:examplea}, the edge weights correspond to some empirical a priori information of the tie strength like contact frequency or duration, e.g., $A$ and $B$ chatted for ten hours and $B$ and $D$ for only one hour. The optimal weighted solution classifies the edges between $A$ and $B$ as well as between $C$ and $D$ as strong (highlighted in red). 
However, if we ignore the weights, as shown in \Cref{fig:exampleb}, the optimal (non-weighted) solution has three strong edges. 
Even though the non-weighted solution has more strong edges, the weighted version agrees more with our intuition and the empirical a priori knowledge. %

\cite{sintos2014using} also introduced a variant of the STC called STC+ that allows adding new weak edges to obtain improved solutions. Similarly to the standard variant, we introduce a weighted version of the STC+.

We employ these generalizations of the STC and the STC+ to infer the strength of ties between nodes in temporal networks.
A temporal network consists of a fixed set of vertices and a chronologically ordered stream of appearing and disappearing temporal edges, i.e., each temporal edge is only available at a specific discrete point in time~\citep{holme2012temporal,oettershagen2022temporalphd}.
Temporal networks can naturally be used as models for real-life scenarios, e.g., communication~\citep{candia2008uncovering,eckmann2004entropy}, contact~\citep{ciaperoni2020relevance,oettershagen2020classifying}, and social networks~\citep{hanneke2006discrete,holme2004structure,moinet2015burstiness}.
In contrast to static graphs, temporal networks are not simple in the sense that between each pair of nodes, there can be several temporal edges, each corresponding to, e.g., a contact or communication at a specific time~\citep{holme2012temporal,oettershagen2022temporalphd}.
Hence, there is no one-to-one mapping between edges and ties. 
Given a temporal network, we map it to a weighted static graph such that the edge weights are a function of the empirical tie strength. We then classify the edges using the weighted STC or STC+, respecting the a priori information given by the edge weights.

A major challenge is that the weighted STC and STC+ are hard to compute, and real-world temporal networks are often provided as large or possibly infinite streams of graph updates.  
To tackle this computational challenge, we employ a sliding time window approach and introduce a streaming algorithm that can efficiently update a 2-approximation of the minimum weighted STC, i.e., the problem that asks for the minimum number of weak edges.
In the case of the STC+, our streaming algorithm efficiently updates a 3-approximation of the minimum weighted STC+.

Note that we infer the tie strength based on the topology of the network. There are various studies on predicting the strength of ties given other features of a network.
\cite{GilbertK09} developed a predictive model to characterize ties in social networks as strong or weak with high accuracy by taking user similarities and interactions into account. 
In the same direction, \cite{XiangNR10} gave an unsupervised model to infer relationship strength based on user similarity and interaction activity.
Moreover, \cite{PhamSL16} showed that spatio-temporal features of social interactions could increase the accuracy of inferred ties strength. However, these works do not classify edges with respect to the STC. 
In contrast, our work is based on the STC property, which was introduced by \cite{granovetter1973strength}.
An extensive analysis of the STC can be found in the book of \cite{EasleyK2010}.
\cite{sintos2014using} not only introduced the optimization problem by characterizing the edges of the network as strong or weak using only the structure of the network, but they also proved that the problem of maximizing the strong edges is NP-hard, and provided two approximation algorithms to solve the dual problem of minimizing the weak edges.
In the following works, the authors of~\cite{GruttemeierK20,KonstantinidisNP18,KonstantinidisP20} focused on restricted networks to further explore the complexity of STC maximization.
\cite{RozenshteinTG2017} discuss the STC with additional community connectivity constraints.
\cite{adriaens2020relaxing} proposed integer linear programming formulations and corresponding relaxations.
Very recently, \cite{MatakosG22} proposed a new problem that uses the strong ties of the network to add new edges and increase its connectivity.
\cite{veldt2022correlation} presented connections between the  cluster editing problem and the STC+. In the cluster editing problem, given a undirected graph $G$, the task is to find the minimum number of edges that need to be introduced to or deleted from $G$ to obtain a disjoint union of cliques.
The author showed that an $\alpha$-approximation for STC+ leads to a $2\alpha$-approximation for the cluster editing problem.

The mentioned works only consider static networks and do not include edge weights in the computation of the STC or STC+. We propose weighted variants and use them to infer ties strength in temporal networks.

Even though temporal networks are a quite recent research field, there are some comprehensive surveys that introduce the notation, terminology, and applications~\citep{holme2012temporal,LatapyVM18,Michail16}.
Recently, there has been a growing interest within the knowledge discovery and data mining community in exploring the field of temporal networks, primarily due to their potential in providing deeper insights into the dynamics and evolution of complex systems, see e.g., \cite{rozenshtein2019mining,santoro2022onbra,oettershagen2022temporal,DBLP:conf/kdd/OettershagenKM23,oettershagen2022tglib}.
Additionally, there are systematic studies into the complexity of well-known graph problems on temporal networks (e.g.~\cite{HimmelMNS17,KempeKK02,ViardLM16}).
The problem of finding communities and clusters, which can be considered as a related problem, has been studied on temporal networks~\citep{ChenMSS18,TantipathananandhB11}. 
Furthermore, \cite{ZhouYR0Z18} studied dynamic network embedding based on a model of the triadic closure process, i.e., how open triads evolve into closed triads. 
\cite{HuangTWLf14} studied the formation of closed triads in dynamic networks.
The authors of~\cite{ahmadian2020theoretical} introduce a probabilistic model for dynamic
graphs based on the triadic closure.

Finally, \cite{wei_et_al} introduced a dynamic $(2 +\varepsilon)$-approximation for the minimum weight vertex cover problem with $\mathcal{O}(\log n / \varepsilon^2)$ amortized update time based on a vertex partitioning scheme~\citep{bhattacharya2018deterministic}.
However, the algorithm does not support updates of the vertex weights, which is an essential operation in our streaming algorithm.

\smallskip\noindent\textbf{Our Contributions:} 
\begin{enumerate}
\item We generalize the STC for weighted graphs and apply the weighted STC for determining tie strength in temporal networks. To this end, we use temporal information to infer the edge strengths of the underlying static graph.
Similarly, we generalize the STC+ variant for weighted graphs that allows insertions of new weak edges.
\item We provide a streaming algorithm framework to efficiently approximate the weighted STC and STC+ over time with an approximation factor of two and three, respectively. As a technical contribution, we propose an efficient dynamic $k$-approximation for the minimum weighted vertex cover problem (MWVC) in $k$-uniform hypergraphs, a key ingredient of our streaming framework. 
\item Our evaluation using real-world temporal networks shows that the weighted STC and STC+ lead to strong edges with higher weights consistent with the given empirical edge weights.
Furthermore, we show the efficiency of our streaming algorithm, which is orders of magnitude faster than the baseline while keeping the same solution quality.
\end{enumerate}

\smallskip
\noindent
In contrast to the conference version~\citep{oettershageninferring}, this work contains all previously omitted details and proofs. Furthermore, we include the following additions:
\begin{enumerate}
\item \textbf{Weighted STC+:} We additionally introduce the weighted STC+ problem and extend our discussion and solution for it.
\item \textbf{Dynamic $k$-Approximation for MWVC:} We lifted the fully dynamic 2-approximation of the minimum weight vertex cover to a fully dynamic $k$-approximation in $k$-uniform hypergraphs. 
\item \textbf{Extended experiments:} We provide additional experiments, including the evaluation of the newly introduced weighted STC+. 
\end{enumerate}

\medskip
\noindent
The remainder of this paper is organized as follows. In \Cref{sec:preliminaries}, we introduce the preliminaries and definitions. \Cref{sec:weightedSTC} presents the generalization of the STC and the STC+ for weighted networks. Next, in \Cref{sec:temporalstc}, we discuss the application of the weighted STC and STC+ for edge strength inference in temporal networks. Furthermore, we introduce our new streaming algorithm, including our dynamic minimum weight vertex cover algorithm.
In \Cref{sec:experiments}, we evaluate  our new approaches, and finally, in \Cref{sec:conclusion}, we provide some concluding remarks and discuss possible directions for future work.

\section{Preliminaries}\label{sec:preliminaries}
An undirected \emph{hypergraph} $H=(V,E)$ consists of a finite set of nodes $V$ and a finite set of \emph{hyperedges} $E\subseteq 2^V\setminus\emptyset$, i.e., each hyperedge connects a non-empty subset of $V$. 
A hypergraph with all its hyperedges of size $k$ is called $k$-uniform hypergraph. 
In particular, we consider the two special cases $k\in\{2,3\}$. %
In the first case, the hypergraph is a conventional undirected graph which we may denote with $G$ and for which we call the hyperedges just edges. 
In the second case, each hyperedge is an unordered triple. 
We use $V(H)$ and $E(H)$ to denote the sets of nodes and hyperedges, respectively, of $H$. 
The set $\delta(v)=\{e\mid e\in E(H), v\in e\}$ contains all hyperedges incident to $v\in V(H)$, and we use $d(v)=|\delta(v)|$ to denote the degree of $v\in V$.
An \emph{edge-weighted} undirected hypergraph $H=(V,E,w_E)$ is an undirected hypergraph with additional weight function $w_E:E\rightarrow\mathbb{R}$.
Analogously, we define a \emph{vertex-weighted} undirected hypergraph $H=(V,E,w_V)$ with a weight function for the vertices $w_V:V\rightarrow\mathbb{R}$. If the context is clear, we omit the subscript of the weight function.

For $2$-uniform hypergraphs, i.e., undirected graphs, we define \emph{wedges}.
A \emph{wedge} is defined as a triplet of nodes $u,v,w \in V$ such that $\{\{u,v\},\{v,w\}\}\subseteq E$ and $\{u,w\}\notin E$. We denote such a wedge by $(v,\{u,w\})$, and with $\mathcal{W}(G)$, the set of wedges in a graph $G$.
Next, we define the \emph{weighted wedge graph}. 
The non-weighted version is also known as the \emph{Gallai} graph~\citep{gallai1967transitiv}.
\begin{definition}
Let $G=(V,E,w_E)$ be an edge-weighted graph.
The \emph{weighted wedge graph} $W(G)=(U,H,w_V)$ consists of the vertex set $U=\{n_{uv}\mid \{u,v\}\in E \}$, the edges set $H=\{ \{n_{uv},n_{vw}\} \mid (v,\{u,w\})\in \mathcal{W}(G) \}$, and the vertex weight function $w_V(n_{uv})=w_E(\{u,v\})$.
\end{definition}
\noindent
\textbf{Temporal Networks~}
A \emph{temporal network} $\tg=(V, \tge)$ consists of a finite set of nodes $V$, a possibly infinite set $\tge$ of undirected \emph{temporal edges} $e=(\{u,v\},t)$ with $u$ and $v$ in $V$, $u\neq v$, and \emph{availability time} (or \emph{timestamp}) $t \in \mathbb{N}$.  For ease of notation, we may denote temporal edges $(\{u,v\},t)$ with $(u,v,t)$. We use $t(e)$ to denote the availability time of $e$. We do not include a duration in the definition of temporal edges, but our approaches can easily be adapted for temporal edges with duration parameters.
We define the underlying static, weighted, \emph{aggregated} graph $A_\phi(\tg)=(V,E,w)$ of a temporal network $\tg=(V,\tge)$ with the edges set $E=\{ \{u,v\} \mid (\{u,v\}, t)\in \tge  \}$ and edge weight function $w:E\rightarrow \mathbb{R}$. 
The edge weights are given by the function $\phi:2^\tge\rightarrow \mathbb{R}$ such that $w(e)=\phi(\mathcal{F}_e)$ with $\mathcal{F}_e=\{e\mid (e,t)\in\tge \}$.
We discuss various weighting functions in \Cref{sec:weighting_functions}.
Finally, we denote the lifetime of a temporal network $\tg=(V,\tge)$ with $T(\tg)=[t_{min},t_{max}]$ with $t_{min}=\min\{t\mid e=(u,v,t)\in \tge\}$ and $t_{max}=\max\{t\mid e=(u,v,t)\in \tge\}$.

\smallskip
\noindent
\textbf{Strong Triadic Closure~}
Given a (static) graph $G = (V,E)$, we can assign one of the labels \emph{weak} or \emph{strong} to each edge in $e\in E$. We call such a labeling a \emph{strong-weak labeling}, and we specify the labeling by a subset $S\subseteq E$. 
Each edge $e\in S$ is called \emph{strong}, and $e\in E\setminus S$ \emph{weak}.
The \emph{strong triadic closure (STC)} of a graph $G$ is a strong-weak labeling $S\subseteq E$ such that
for any two strong edges $\{u,v\}\in S$ and $\{v,w\}\in S$, there is a (weak or strong) edge $\{u,w\}\in E$.
We say that such a labeling \emph{fulfills} the strong triadic closure.
In other words, in a strong triadic closure, there is no pair of strong edges $\{u,v\}$ and $\{v,w\}$ such that $\{u,w\} \notin E$.
Consequently, a labeling $S\subseteq E$ fulfills the STC if and only if at most one edge of any wedge in $\mathcal{W}(G)$ is in $S$, i.e., there is no wedge with two strong edges~\citep{sintos2014using}.
The decision problem for the STC is denoted by \textsc{MaxSTC} and is stated as follows:
Given a graph $G = (V,E)$ and a non-negative integer $k$.
Does there exist $S\subseteq E$ that fulfills the strong triadic closure and $|S|\geq k$?

Equivalently, we can define the problem based on the weak edges, \textsc{MinSTC}, in which given a graph $G = (V,E)$ and a non-negative integer $\ell$.
Does there exist $E'\subseteq E$ that $E\setminus E'$ fulfills the strong triadic closure and $|E'|\leq \ell$?

\smallskip
\noindent
\textbf{Strong Triadic Closure with Edge Additions~}
Here apart from labeling the edges of the graph as strong or weak, we can add new edges between non-adjacent nodes and label these edges as weak. We denote this problem by \textsc{MinSTC+} and it is stated as follows:
Given a graph $G = (V,E)$ and a non-negative integer $\ell$. 
Does there exist a set $F\subseteq {V \choose 2}\setminus E$ such that there is a $E'\subseteq E$ that $E\setminus E'$ fulfills the strong triadic closure and $|E' \cup F|\leq \ell$?

The motivation for considering edge additions is the following~\citep{sintos2014using}:
Let $G$ be an (almost) complete graph on $n$ vertices with exactly one edge $\{u,v\}$ missing. In this case, the best strong-weak labeling for $G$ contains $n-2$ weak edges. By adding the single edge $\{u,v\}$ to $G$, we obtain the complete graph over $n$ vertices for which all edges can be labeled strong. 
Hence, by adding a single edge, we are able to improve our labeling hugely. 

\subsection{Examples of Temporal, Aggregated, and Wedge Graph}
\Cref{fig:tgexample} shows from left to right (i) a temporal graph $\tg$ and (ii) its aggregated graph $A_\phi(\tg)$ with $\phi$ being the weighting function according to contact frequency. Furthermore, (iii) the wedge graph $W(A(\tg))$ with highlighted minimum weight vertex cover. The blue nodes are the vertices in the cover. Finally, (iv) the corresponding minimum weight STC labeling of the edges in $A(\tg)$ is shown. The red edges are strong, and the blue dashed ones are weak. 

\begin{figure}[htb]
    \centering
    \includegraphics{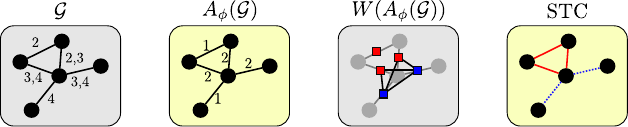}
    \caption{Example for temporal graph $\tg$, the aggregated graph $A_\phi(\tg)$, the wedge graph $W(A_\phi(\tg))$, and minimum weight STC.}
    \label{fig:tgexample}
\end{figure}

\section{Weighted Strong Triadic Closure}\label{sec:weightedSTC}

Let $G=(V,E,w)$ be a graph with edge weights reflecting the importance of the edges. We determine a \emph{weighted} strong triadic closure that takes into account the weights of the edges by the importance given by $w$. To this end, let $S\subseteq E$ be a strong-weak labeling. The labeling $S$ fulfills the weighted STC if (1) for any two strong edges $\{u,v\},\{v,w\}\in S$ there is a (weak or strong) edge $\{u,w\}\in E$, i.e., fulfills the unweighted STC, and (2) maximizes $\sum_{e\in S}w(e)$.

The corresponding decision problem \textsc{WeightedMaxSTC} has as input a graph $G=(V,E)$ and $U\in \mathbb{R}$, and asks for the existence of a strong-weak labeling that fulfills the strong triadic closure and for which $\sum_{e\in S}w(e)\geq U$.
\cite{sintos2014using} showed that \textsc{MaxSTC} is NP-complete using a reduction from Maximum Clique.
The reduction implies that we cannot approximate the \textsc{MaxSTC} with a factor better than $\mathcal{O}(n^{1-\epsilon})$. 
Because \textsc{MaxSTC} is a special case of \textsc{WeightedMaxSTC}, these negative results also hold for the latter. 

Instead of maximizing the weight of strong edges, we can equivalently minimize the weight of weak edges resulting in the corresponding problem \textsc{WeightedMinSTC}\footnote{We use \textsc{WeightedMinSTC} for the decision and the optimization problem in the following if the context is clear.}. 
Here, we search a strong-weak labeling that fulfills the STC and minimizes the weight of the edges not in $S$.
Both the maximation and the minimization problems can be solved exactly using integer linear programming (ILP).
We provide the corresponding ILP formulations in \Cref{appendix:ilp}.
The advantage of \textsc{WeightedMinSTC} is that we can obtain a 2-approximation.

To approximate \textsc{WeightedMinSTC} in an edge-weighted graph $G=(V,E,w)$, we first construct the weighted wedge graph $W(G)=(V_W,E_W, w_{V_W})$.
Solving the minimum weighted vertex cover (MWVC) problem on $W(G)$ leads then to a solution for the minimum weighted STC of $G$, where MWVC is defined as follows. Given a vertex-weighted graph $G=(V, E, w)$, the minimum weighted vertex cover asks if there exists a subset of the vertices $C\subseteq V$ such that each edge $e\in E$ is incident to a vertex $v\in C$ and the sum $\sum_{v\in C}w(v)$ is minimal.
\begin{lemma}\label{lemma:mwvcstc}
    Solving the MWVC on $W(G)$ leads to a solution of the minimum weight STC on $G$.
\end{lemma}
\begin{proof}
    It is known that a (non-weighted) vertex cover $C\subseteq V(W)$ in $W(G)$ is in one-to-one correspondence to a (non-weighted) STC in $G$, see~\cite{sintos2014using}.
    The idea is the following. 
    Recall that the wedge graph $W(G)$ contains for each edge $\{i,j\}\in E(G)$ one vertex $n_{ij}\in V(W)$. Two vertices $n_{uv}, n_{uw}$ in $W(G)$ are only adjacent if there exists a wedge $(u,\{v,w\})\in \mathcal{W}(G)$. 
    If we choose the weak edges $E\setminus S$ to be the edges $\{i,j\}\in E$ such that $n_{ij}\in C$, each wedge has at least one weak edge.
    Now for the weighted case, by definition, the weight of the STC $\sum_{e\in E\setminus S}w(e)$ equals the weight of a minimum vertex cover $\sum_{v\in C}w_{V_W}(v)$. 
\end{proof}

\Cref{lemma:mwvcstc} implies that an approximation for \textsc{MWVC} yields an approximation for \textsc{WeightedMinSTC}.
A well-known 2-approximation for the MWVC is the pricing method which is based on the primal-dual method and was first suggested in~\cite{bar1981linear}.
The idea of the pricing algorithm is to assign to each edge $e\in E$ a {price} $p(e)$ initialized with zero. We say a vertex is \emph{tight} if the sum of the prices of its incident edges equals the weight of the vertex. 
We iterate over the edges, and if for $e=\{u,v\}$ both $u$ and $v$ are not tight, 
we increase the price of $p(e)$ until at least one of $u$ or $v$ is tight. 
Finally, the vertex cover is the set of tight vertices.
See, e.g.,~\cite{kleinberg2006algorithm} for a detailed introduction.
In \Cref{sec:stcovertime}, we generalize the pricing algorithm for fully dynamic updates of edge insertions and deletions and vertex weight updates.

\subsection{Weighted STC+ Variant}\label{sec:stcplus}
As above, let $G=(V,E,w)$ be a graph with edge weights reflecting the importance of the edges. We consider the weighted STC+ variant that allows insertions of weak edges called \textsc{WeightedMinSTC+}.
Our goal is to find a subset of new edges $E_N\subseteq {V \choose 2}\setminus E$ and a strong-weak labeling $S$ such that 
(1) $S$ fulfills the STC in $G=(V, E\cup E_N)$ and 
(2) the weight $\sum_{e\in E_N\cup (E\setminus S)}w(e)$ is minimized.
The new edges in $E_N$ are considered to be weak; thus, no new open wedges that violate strong triadic closure are introduced.
We still need to assign the weight $w(f)$ to the new edges $f\in E_N$. 
Each edge $f=\{u,w\}$ can be part of a set of wedges $\mathcal{W}_f=\{(v,\{u,w\})\mid v\in V\}\subseteq \mathcal{W}(G)$.
We set the weight of $f=\{u,w\}$ to
\[
w(f)=\alpha \cdot \sum_{(v,\{u,w\})\in \mathcal{W}_f} \frac{w(\{v,u\})+w(\{v,w\})}{|\mathcal{W}_f|}\text{,}
\]
with the parameter $\alpha\in\mathbb{R}_{>0}$, which can be used to allow more or fewer new edges. Note that for unit weights, i.e., in the unweighted STC+ case, we obtain the weight of $w(f)=1$ using  $\alpha = \frac{1}{2}$. Furthermore, for a large enough $\alpha$, e.g., $\alpha\geq \sum_{(v,\{u,w\})\in \mathcal{W}}(w(\{v,u\})+w(\{v,w\}))$, the the solution of the weighted STC+ does not contain any edge $f\in E_N$ and is a solution for the \textsc{WeightedMinSTC} problem.

Similarly to \textsc{WeightedMinSTC}, we provide an approximation based on the minimum vertex cover problem. To this end, we construct a \emph{vertex-weighted hypergraph} $H_W=(V_W,E_W,w_{V_W})$ that contains a hyperedge for each wedge in $G=(V,E,w_E)$.
Particularly, we define 
\begin{align*}
    V_W&=\{n_{uv}\mid \{u,v\}\in E \}\cup \{n'_{xy}\mid \{x,y\}\in E_N\subseteq \textstyle{V \choose 2} \setminus E \}\text{, and}\\
    E_W&=\{ \{n_{uv},n_{vw},n'_{uw}\} \mid (v,\{u,w\})\in \mathcal{W}(G) \}.
\end{align*}
and the vertex weight function $w_V(n_{uv})=w(\{u,v\})$.
Now, a minimum weight vertex cover in $H_W$ leads to an optimal solution of \textsc{WeightedMinSTC+}. 
Let $C\subseteq V_W$ be a vertex cover in $H_W$, i.e., for each hyperedge $e\in E_W$, there exists a vertex $n_{uv}\in e$ with $n_{uv}$ also in $C$, and the corresponding edge $\{u,v\}\in E\cup E_N$ is labeled weak. Consequently, in case $\{u,v\}\in E$, the to $e$ corresponding wedge has a weak edge. And if $\{u,v\}\in E_N$, the to $e$ corresponding wedge is closed by the new edge $\{u,v\}$. Hence, $C$ leads to a valid labeling for $G$. Because $C$ has minimum weight, the corresponding labeling is optimal.

We can obtain a 3-approximation algorithm using the pricing method analogously to the approximation for \textsc{WeightedMinSTC}. 
Note that this improves the $\mathcal{O}(\log n)$ approximation stated in \cite{sintos2014using}.

\subsection{ILP Formulations}\label{appendix:ilp}
For the exact computation of the weighted STC, we provide an ILP formulation similar to the  formulation of the unweighted version of the maximum STC problem~\citep{adriaens2020relaxing}.
We use binary variables $x_{ij}$ which encode the if edge $\{i,j\}\in E$ is strong $(x_{ij}=1)$ or weak $(x_{ij}=0)$.
Moreover, $w_{ij}$ is the weight of edge $\{i,j\}$.
For \textsc{MaxWeightSTC}, we have
\begin{align}
    &\max \sum_{ij\in E}w_{ij}x_{ij}\\
    &\text{s.t.} \quad x_{ij} + x_{ik} \leq 1 \quad \text{for all }(i,\{j,k\})\in \mathcal W(G)\\
    &\quad\quad x_{ij}\in\{0,1\}\quad \text{for all }\{i,j\} \in E.
\end{align}
Similarly, we define the minimization version for \textsc{MinWeightSTC}, where we use binary variables $y_{ij}$ which encode the if edge $\{i,j\}\in E$ is strong $(y_{ij}=0)$ or weak $(y_{ij}=1)$, as
\begin{align}
    &\min \sum_{ij\in E}w_{ij}y_{ij}\\
    &\text{s.t.} \quad y_{ij} + y_{ik} \geq 1 \quad \text{for all }(i,\{j,k\})\in \mathcal W(G)\\
    &\quad\quad y_{ij}\in\{0,1\}\quad \text{for all }\{i,j\} \in E.
\end{align}

Finally, in the case of the \textsc{WeightedMinSTC+} problem, we adapt the ILP formulation stated for the unweighted version of the minimum STC+ problem~\citep{veldt2022correlation}. We use binary variables $y_{ij}$ which encode the if (possibly new) edge $\{i,j\}\in E\cup E_N$ is strong $(y_{ij}=0)$ or weak $(y_{ij}=1)$.
Hence, for \textsc{WeightedMinSTC+}, we have
\begin{align}
    &\min \sum_{ij\in {V \choose 2}}w_{ij}y_{ij}\\
    &\text{s.t.} \quad y_{ij} + y_{ik} + y_{jk} \geq 1 \quad \text{for all }(i,\{j,k\})\in \mathcal W(G)\\
    &\quad\quad y_{ij}\in\{0,1\}\quad \text{for all }\{i,j\} \in \textstyle{V \choose 2}.
\end{align}

\section{Strong Triadic Closure in Temporal Networks}\label{sec:temporalstc}
We first present meaningful weighting functions to obtain an edge-weighted aggregated graph from the temporal network.
Next, we discuss the approximation of the \textsc{WeightedMinSTC} and the \textsc{WeightedMinSTC+} in the non-streaming case. 
Finally, we introduce the approximation streaming algorithms for temporal networks.

\subsection{Weighting Functions for the Aggregated Graph}\label{sec:weighting_functions}
A key step in the computation of the STC for temporal networks is the aggregation and weighting of the temporal network to obtain a weighted static network.
Recall that the weighting of the aggregated graph $A_\phi(\tg)$ is determined by the weighting function $\phi:2^\tge\rightarrow \mathbb{R}$ such that $w(e)=\phi(\mathcal{F}_e)$ with $\mathcal{F}_e=\{e\mid (e,t)\in\tge \}$. 
Naturally, the weighting function $\phi$ needs to be meaningful in terms of tie strength; hence, we propose the following variants:
\begin{itemize}
    \item \emph{Contact frequency:} We set $\phi(\mathcal{F}_e)=|\mathcal{F}_e|$, i.e., the weight $w(e)$ of edge $e$ in the aggregated graph equals the number of temporal edges between the endpoints of $e$.
    Contact frequency is a popular and common substitute for tie strength~\citep{GilbertK09,granovetter1973strength,lin1978analyzing}.
    \item \emph{Exponential decay:} \cite{lin1978analyzing}  proposed to measure tie strength in terms of the recency of contacts. 
    We propose the following weighting variant to capture this property where $\phi(\mathcal{F}_e)=\sum_{i=1}^{|\mathcal{F}_e|-1}e^{-(t(e_{i+1})-t(e_i))}$ if $|\mathcal{F}_e|\geq 2$ and else $\phi(\mathcal{F}_e)=0$. Here, we interpret $\mathcal{F}_e$ as a chronologically ordered sequence of the edges. 
    \item \emph{Duration:} Temporal networks can include durations as an additional parameter of the temporal edges, i.e., each temporal edge $e$ has an assigned value $\lambda(e)\in\mathbb{N}$ that describes, e.g., the duration of a contact~\citep{holme2012temporal}. The duration is also commonly used as an indicator for tie strength~\citep{GilbertK09}.
    We can define $\phi$ in terms of the duration, e.g., $\phi(\mathcal{F}_e)=\sum_{f\in\mathcal{F}_e}\lambda(f)$.
\end{itemize}
Other weighting functions are possible, e.g., combinations of the ones above or weighting functions that include node feature similarities.

\subsection{Approximation of \textsc{WeightedMinSTC} and \textsc{WeightedMinSTC+}}
Before introducing our streaming algorithm, we discuss how to compute and approximate the \textsc{WeightedMinSTC} and \textsc{WeightedMinSTC+} in a temporal network $\tg=(V,\tge)$ in the non-streaming case. Consider the following algorithm for the \textsc{WeightedMinSTC} case:
\begin{enumerate}
    \item Compute $A_\phi(\tg)=(V,E, w)$ using an appropriate weighting function $\phi$.
    \item Compute the vertex-weighted wedge graph $W(A_\phi(\tg))=(V_W, E_W, w_{V_W})$. %
    \item Compute an MWVC $C$ on $W(A_\phi(\tg))$.
\end{enumerate}
The nodes $n_{uv}$ in $C$ then correspond to the weak ties $\{u,v\}$ in $\tg$.
For the \textsc{WeightedMinSTC+}, we can compute in step 2 the vertex-weighted wedge hypergraph as described in~\Cref{sec:stcplus}.
Depending on how we solve step three, we can either compute an optimal or approximate solution, e.g., using the pricing approximation for the MWVC, we obtain a 2-approximation for \textsc{WeightedMinSTC} or a 3-approximation for \textsc{WeightedMinSTC+}, respectively.
Using the pricing approximation, we have a linear running time in the number of (hyper-)edges in the wedge graph. 
The problem with this direct approach of computing a solution for the \textsc{WeightedMinSTC} or \textsc{WeightedMinSTC+} problem is its limited scalability.
The reason is that the number of vertices in the wedge graph $|V_W|\geq|E(A_\phi(\tg))|$ and the number of edges equals the number of wedges in $A$, which is bounded by $\mathcal{O}(|V|^3)$, see~\cite{pyatkin2019maximum}, leading to a total running time and space complexity of $\mathcal{O}(|V|^3)$.

\subsection{Streaming Framework for \textsc{WeightedMinSTC} and \textsc{WeightedMinSTC+}}\label{sec:stcovertime}
In the previous section, we saw that the size of the wedge graph could render the direct approximation approach infeasible for large temporal networks.
To overcome this obstacle, we use a sliding time window of size $\Delta\in\mathbb{N}$ to compute the changing STC or STC+ for each time window.
The advantage is two-fold: (1) By considering limited time windows, the size of the wedge graphs for which we have to compute the MWVC is reduced because, usually, not all participants in a network have contact in the same time window. 
(2) If we consider temporal networks spanning a long (possibly infinite) time range, the relationships, and thus, tie strengths, between participants change over time. Using the sliding time window approach, we can capture such changes.

In the following discussion, we first explain our streaming algorithm for the STC problem, and later in \Cref{sec:stcplusstreaming}, we describe the necessary adaptions for the STC+ variant.
Moreover, we assume the weighting function $\phi$ to be linear in the contact frequency, and we omit the subscript. But, our results are general and can be applied to other weighting functions.
Let $\tau$ be a time interval and let $A(\tg(\tau))$ be the aggregated graph of $\tg(\tau)$, i.e., the temporal network that only contains edges starting and arriving during the interval $\tau$.
For a time window size of $\Delta\in\mathbb{N}$, we define the sliding time window $\tau_t$ at timestamp $t$ with $t \in [1, T(\tg) - \Delta+1]$ as $\tau_t = [t, t + \Delta-1]$. 

\begin{figure}[t]
    \centering
    \includegraphics[width=1.0\linewidth]{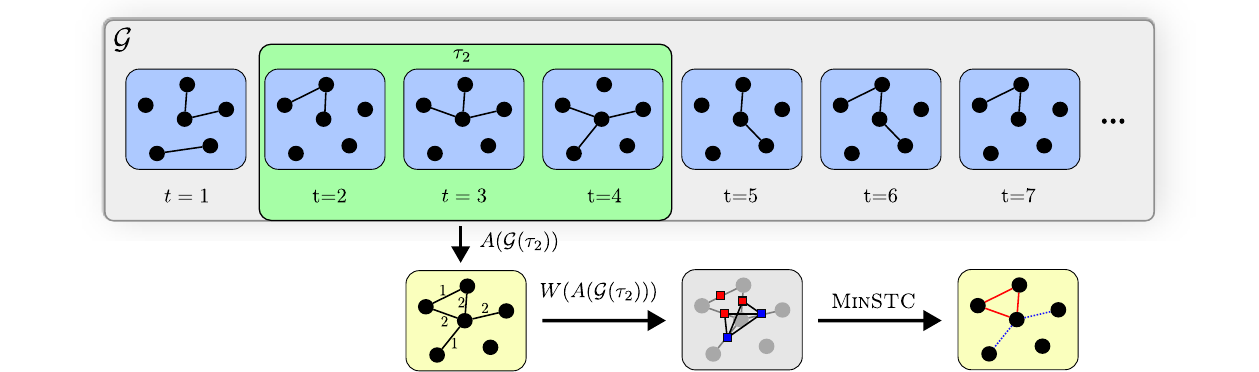}
    \caption{Example for computing the weighted STC of a sliding time window.}
    \label{fig:streaming}
\end{figure}

\Cref{fig:streaming} shows an example of our streaming approach for $\Delta=3$.
The first seven timestamps of temporal network $\tg$ are shown as static slices.
The time window $\tau_2$ of size three starts at $t=2$. First, the static graph $A(\tg(\tau_2))$ is aggregated, and the wedge graph $W(A(\tg(\tau_2)))$ is constructed. The wedge graph is used to compute the weighted STC. After this, the time window is moved one time stamp further, i.e., it starts at $t=3$ and ends at $t=5$, and the aggregation and STC computation are repeated (not shown in \Cref{fig:streaming}).
In the following, we describe how the aggregated and wedge graphs are updated when the time window is moved forward, how the MWVC is updated for the changes of the wedge graph, and how the final streaming algorithm proceeds.

\subsubsection{Updating the aggregated and wedge graphs}\label{sec:streaming}
Let $\tau_{t_1}$ and $\tau_{t_2}$ be to consecutive time windows, i.e., $t_2=t_1+1$. Furthermore, let $A_i=A(\tg(\tau_{t_i}))$ and $W_i=W(A(\tg(\tau_{t_i})))$ with $i\in\{1,2\}$ be the corresponding aggregated and wedge graphs.
The sets of edges appearing in the time windows $\tg(\tau_{t_1})$ and $\tg(\tau_{t_2})$ might differ. For each temporal edge that is in $\tg(\tau_{t_1})$ but not in $\tg(\tau_{t_2})$, we reduce the weight of the corresponding edge in the aggregated graph $A_1$. If the weight reaches zero, we delete the edge from $A_1$. Analogously, for each temporal edge that is in $\tg(\tau_{t_2})$ but not in $\tg(\tau_{t_1})$, we increase the weight of the corresponding edge in $A_1$. If the edge is missing, we insert it. This way, we obtain $A_2$ from $A_1$ by a sequence of update operations. 
Now, we map these edge removals, additions, and edge weight changes between $A_1$ and $A_2$ to updates on $W_1$ to obtain $W_2$. 
For each edge removal (addition) $e=\{u,v\}$ between $A_1$ and $A_2$, we remove (add) the corresponding vertex (and incident edges) in $W_1$. We also have to add or remove edges in $W_1$ depending on newly created or removed wedges. More precisely, for every new wedge in $A_1$, we add an edge between the corresponding vertices in $W_1$, and for each removed wedge, i.e., by deleting an edge or creating a new triangle, in $A_1$, we remove the edges between the corresponding vertices in $W_1$.
Furthermore, for each edge weight change between $A_1$ and $A_2$, we decrease (increase) the weight of the corresponding vertex in $W_1$. 
Hence, the wedge graph $W_1$ is edited by a sequence $\sigma$ of vertex and edge insertions, vertex and edge removals, and weight changes to obtain $W_2$.  
Because we only need to insert or remove a vertex in the wedge graph $W_1$ if the degree changes between zero and a positive value, we do not consider vertex insertion and removal in $W_1$ as separate operations in the following.
The number of vertices and edges in $W_1$ is bounded by the current numbers of edges and wedges in $A_1$. Furthermore, we bound the number of changes in $W_1$ after inserting or deleting edges from $A_1$.

\begin{lemma}\label{lemma:numberedges}
    The number of new edges in $W_1$ after inserting $\{v,w\}$ into $A_1$ is at most $d(v)+d(w)$, and the number of edges removed from $W$ is at most $\min(d(v),d(w))$.
    The number of new edges in $W_1$ after removing $\{v,w\}$ from $A_1$ is at most $\min(d(v),d(w))$, and the number of edges removed from $W_1$ is at most $d(v)+d(w)$.
\end{lemma}
\begin{proof}
    By adding $\{v,w\}$ in $A$, the number of new wedges in which a vertex $v$ can be part is at most $d(v)$, and so we can at most create $d(v)$ new edges in $E(W)$.
    Analogously, for $w$, we can create at most $d(w)$ new edges.
    Therefore, $v$ and $w$ create at most $d(v)+d(w)$ edges in $E(W)$.
    Analogously, if we remove $\{v,w\}$, we destroy at most $d(v)+d(w)$ edges in $E(W)$. 
    
    The number of triangles that contain edge $\{v,w\}$ is bounded by $\min(d(v),d(w))$. If we add $\{v,w\}$ to $A$, we may close triangles of the form $\{u,v,w\}$. We remove one edge from $E(W)$ for each such triangle that is closed by edge $\{v,w\}$ in $A$; hence, at most $\min(d(v),d(w))$ edges are removed.
    On the other hand, by removing $\{v,w\}$ from $E(A)$, we can destroy at most $\min(d(v),d(w))$ triangles in $A$ and have to insert corresponding edges in $E(W)$. 
\end{proof}

\subsubsection{Updating the MWVC}\label{sec:updatingmwvc}
If the sliding time window moves forward, the current wedge graph $W$ is updated by the sequence $\sigma$. We consider the updates occurring one at a time and maintain a $2$-approximation of an MWVC in $W$.
\Cref{alg:pricing} shows our dynamic pricing approximation based on the non-dynamic approximation for the MWVC. 
The algorithm supports the needed operations of inserting and deleting edges, as well as increasing and decreasing vertex weights.
When called for the first time, an empty vertex cover $C$ and wedge graph $W$ are initialized (line~\ref{alg:init} and following), which will be maintained and updated in subsequent calls of the algorithm.
In the following, we show the general result that our algorithm gives a $k$-approximation of the MWVC after each update operation in a $k$-uniform hypergraph. 
\begin{definition}
    We assign to each edge $e\in E(W)$ a price $p(e)\in \mathbb{R}$.
    We call prices fair, if $s(v)=\sum_{e\in \delta(v)}p(e)\leq w(v)$ for all $v\in V(W)$.
    And, we say a vertex $v\in V(W)$ is tight if $s(v)=w(v)$.
\end{definition}

Let $W$ be the current hypergraph (i.e. the current wedge graph) and $\sigma$ a sequence of dynamic update requests, i.e., inserting or deleting edges and increasing or decreasing vertex weights in $W$. \Cref{alg:pricing} 
calls for each request $r\in \sigma$ a corresponding procedure to update $W$ and the current vertex cover $C$ (line \ref{alg:main_loop} and following).
We show that after each processed request, the following invariant holds.

\begin{invariant}\label{inv:fairness}
    The prices are fair, i.e., $s(v)\leq w(v)$ for all vertices $v\in V(W)$, and $C\subseteq V(W)$ is a vertex cover.
\end{invariant}

\begin{lemma}\label{lemma:dynamic}
    If Invariant \ref{inv:fairness} holds, after calling one of the procedures \textsc{InsEdge}, \textsc{DelEdge}, \textsc{DecWeight}, or \textsc{IncWeight} Invariant \ref{inv:fairness} still holds.
\end{lemma}
In order to show Lemma 3, we first show the following properties for the \textsc{Update} procedure.
\begin{lemma}\label{lemma:update}
    If the prices are fair, after calling $\textsc{Update}$, the prices are still fair. 
    Furthermore, let $C$ cover all hyperedges $E(W)\setminus F$. After $\textsc{Update}$, all hyperedges in $F$ are covered by a vertex in $C$.
\end{lemma}
\begin{proof}
    For each hyperedge $e\in F$ with a tight vertex, $p(e)$ is not changed, and $e$ is covered by the tight vertex in $C$.
    Now, for each hyperedge $e\in F$ with no endpoint being tight, the price $p(e)$ is increased until one $u\in e$ is tight.
    Therefore, $s(u)\leq w(u)$ for all $u\in e$. Furthermore, at least one of the vertices $u\in e$ is added to $C$; hence, $e$ is covered. 
\end{proof}

\begin{proof}[Proof of \Cref{lemma:dynamic}]
    We prove the lemma for each procedure separately.
    \begin{itemize}
        \item \textsc{InsEdge}: Let $e_n$ be the new hyperedge. If $u\in e_n$ is tight, the price of $e_n$ will be $p(e_n)=0$, and hence $s(u)$ for $u\in e_n$ does not change. Moreover, $e_n$ is covered.
        If all $u\in e_n$ are not tight, the \textsc{Update} procedure will add $u\in e_n$ to $C$ while maintaining fairness by \Cref{lemma:update}.
        \item \textsc{DelEdge}: Let $e_d$ be the hyperedge that is removed from $E(W)$. By updating $s(u)$ for $u\in e_d$ to their new values $s'(u)=s(u)-p(e_d)$ the vertices $u\in e_d$ might not be tight anymore. We remove them from $C$, which then might not be a valid vertex cover anymore. However, the prices are with $s'(u)$ fair. 
        To repair the vertex cover $C$, we call \textsc{Update} with all edges incident to $u\in e_d$ whose other endpoints are not tight. These are precisely the hyperedges not covered by $C$. After running \textsc{Update}($F$), $C$ is a vertex cover, and fairness is maintained (\Cref{lemma:update}).
        \item \textsc{DecWeight}: Decreasing the weight of $w(v)$ can affect the fairness. Hence, we have to adapt, i.e., lower, the prices of the hyperedges incident to $v$. By setting the prices of hyperedges $e$ incident to $v$ to $p(e)=0$, fairness is restored. 
        However, this might affect the tightness of neighbors of $v$ because a formerly tight neighbor $x$ connected by a hyperedge $f$ might not be tight after setting $p(f)=0$. Hence, we remove $v$ and neighbors that are not tight anymore from $C$ and collect all non-covered hyperedges from neighbors $x$ of $v$ into the set $F$. %
        The set $F$ contains all edges that are not covered.
        Calling \textsc{Update} on the set $F$ leads to a valid vertex cover $C$ while maintaining fairness (\Cref{lemma:update}).
        
        \item \textsc{IncWeight}: Increasing the weight $w(v)$ does not affect fairness. However, it affects the tightness of $v$. Hence, if $v\in C$, we first remove $v$ from $C$ and add all hyperedges $f$ incident to $v$ for which all $x\in f$ are also not tight to the set $F$.
        All other edges not in $F$ are covered by some $w\in C$.
        Again,  by calling \textsc{Update} on the set $F$, we obtain a valid vertex cover $C$ and maintain fairness (\Cref{lemma:update}).
    \end{itemize}
\end{proof}
\begin{theorem}
    \Cref{alg:pricing} 
    maintains a vertex cover with $w(C)\leq k\cdot w(C^*)$, where $C^*$ is an optimal MWVC.
\end{theorem}
\begin{proof}
    \Cref{lemma:dynamic} ensures that $C$ is a vertex cover and after each dynamic update the prices are fair, i.e., $\sum_{e\in \delta(v)}p(e)\leq w(v)$.
    Furthermore, (1) for an optimal MWVC $C^*$ and fair prices, it holds that $\sum_{e\in E(W)}p(e)\leq w(C^*)$.
    To see this, consider the optimal MWVC $C^*$ and 
    \[
    \sum_{u\in C^*}\sum_{e\in E(W):u\in e} p(e) \leq \sum_{u\in C^*}w(u)=w(C^*).
    \]
    Because $C^*$ is a vertex cover, each edge contributes $p(e)$ (at least once) to the left hand side.
    Hence, 
    \[
    \sum_{e\in E(W)} p(e)\leq\sum_{u\in C^*}\sum_{e\in E(W):u\in e} p(e) \leq \sum_{u\in C^*}w(u)=w(C^*).
    \]
    
    Moreover, (2) for the vertex cover $C$ and the computed prices, it holds that $\frac{1}{k}w(C)\leq \sum_{e\in E(W)}p(e)$. 
    First note that $\sum_{e\in E:u\in e}p(e)=w(u)$ for all $u\in C$ because each node in $C$ is tight.
    Therefore, we have
    \[w(C)=\sum_{u\in C}w(u)=\sum_{u\in C}\sum_{e\in E(W):u\in e}p(e)\leq k\cdot \sum_{e\in E(W)}p(e),\]
    where the last inequality holds because for each $k$-uniform hyperedge $e$, we count $p(e)$ at most $k$ times. %
    
    Consequently, the theorem follows from (1) and (2) as we have $w(C)\leq k\cdot \sum_{e\in E}p(e)$ and $k\cdot \sum_{e\in E} p(e)\leq k\cdot w(C^*)$.
\end{proof}

We now discuss the running times of the dynamic update procedures.
For each of the four operations, the running time is in $\mathcal{O}(F)$, i.e., the size of the set for which we call the \textsc{Update} procedure.
\begin{theorem}
    Let $d_\text{max}$ be the largest degree of any vertex in $V(W)$. The running time of 
    \textsc{InsEdge} is in $\mathcal{O}(1)$, and \textsc{DelEdge} is in $\mathcal{O}(d_\text{max})$.
    \textsc{DecWeight} is in $\mathcal{O}(d_\text{max}^2)$, and \textsc{IncWeight} is in $\mathcal{O}(d_\text{max})$.
\end{theorem}
\begin{algorithm}
    \caption{Dynamic Pricing Approximation }
    \label{alg:pricing}
    \small
    \Input{Sequence $\sigma$ of dynamic update requests}
    \Output{Algorithm maintains a $k$-approximation of MWVC $C$}
    \BlankLine
    \DontPrintSemicolon
    \SetKwProg{Fn}{Procedure}{:}{}

        Initialize and maintain vertex cover $C$\;\label{alg:init}   
        and wedge graph $W$\;
        
        \BlankLine
        \SetKwFunction{FUpdate}{Update}
        \Fn{\FUpdate{set of edges $F$}}{
            \ForEach{$e\in F$}{
                \lIf{there is $u\in e$ that is tight}{
                    continue
                }
                increase $p(e)$ until $u\in e$ is tight\;
            }
            add newly tight vertices to $C$\; 
        }

        \SetKwFunction{FDecreaseWeight}{DecWeight}
        \Fn{\FDecreaseWeight{$v$, $w_n$}}{
            $w(v)\gets w_n$\;
            
            $C\gets C\setminus \{v\}$\;
            
            $F'\gets \delta(v) $ and initialize $F = \emptyset$\;
            
            \ForEach{$e\in F'$}{
                $p(e)\gets 0$\;
                
                \If{all $x\in e$ are not tight}{
                    $C\gets C\setminus e$\;
                    
                    $F\gets F\cup \{f \in E(G) \mid f\cap e\neq \emptyset \text{ and all $x\in f$ are not tight } \} \cup \{e\}$\;
                }
            }
            \FUpdate{$F$}
        }

        \SetKwFunction{FDelete}{DelEdge}
        \Fn{\FDelete{$e_d\in E(W)$}}{
            $E(W)\gets E(W)\setminus \{e_d\}$\;
            
            update $s(u)$ for $u\in e_d$\;
            
            $C\gets C\setminus e_d$\;
            
            $F\gets \{ f \in E(W)\mid e_d\cap f \neq \emptyset \text{ and all } x\in f \text{ are not tight} \} $\;
            
            \FUpdate{$F$}
        }

        \SetKwFunction{FInsert}{InsEdge}
        \Fn{\FInsert{$e_n$}}{
            $E(W)\gets E(W)\cup \{e_n\}$\;
            
            \FUpdate{$\{e_n\}$}
        }

        \SetKwFunction{FIncreaseWeight}{IncWeight}
        \Fn{\FIncreaseWeight{$v$, $w_n$}}{
            $w(v)\gets w_n$\;
            
            \If{$v \in C$}{
                $C\gets C\setminus \{v\}$\;
                
                $F\gets \{ e \in E(W)\mid v\in e \text{ and there is no tight vertex in } e\} $\;
                
                \FUpdate{$F$}
            }
        }
        \BlankLine  
        
        \ForEach{update request $r\in \sigma$} {\label{alg:main_loop}
            Call the to $r$ corresponding procedure\\ $f\in \{\FInsert, \FDelete,$ $ \FIncreaseWeight, \FDecreaseWeight\}$.
        }
\end{algorithm}

\subsubsection{The streaming algorithm for {MinWeightSTC}}\label{sec:streamfinal}
\Cref{alg:streaming} 
shows the final streaming algorithm that expects as input a stream of chronologically ordered temporal edges and the time window size $\Delta$. As long as edges are arriving, it iteratively updates the time windows and uses \Cref{alg:pricing} 
to compute the \textsc{MinWeightSTC} approximation for the current time window $\tau_t$ with $t\in[1,T(\tg)-\Delta]$.
\Cref{alg:streaming} 
outputs the strong edges based on the computed vertex cover $C_t$ in line~\ref{alg:streaming:output}.
It skips lines \ref{alg:streaming:computeA}-\ref{alg:streaming:output} if there are no changes in $E_\tau$.
\begin{theorem}
    Let $d^W_t$ ($d^A_t$) be the maximal degree in $W$ ($A$, resp.) after iteration $t$ of the while loop in \Cref{alg:streaming}. 
    The running time of iteration $t$ is in $\mathcal{O}(\xi\cdot d_t^A\cdot (d_t^W)^2)$, with $\xi=\max\{|E^-_t|, |E^+_t|\}$.
\end{theorem}
\begin{proof}
    We have the worst-case running time for a sequence $\sigma_t$ that contains only \textsc{DecWeight} requests; see Theorem~2.
    Hence, the running time is in $\mathcal{O}(|\sigma_t|\cdot (d_t^W)^2)$.
    For the length of the sequence $\sigma_t$, we have the following considerations.
    By Lemma 2, we know that the number of \textsc{InsEdge} and \textsc{DelEdge} requests for one new or removed edge are in $\mathcal{O}(d_t^A)$, where $d_t^A$ is the maximal degree in $A$. 
    So any edge insertion into (or removal from) $A$ during iteration $t$ leads to at most $\mathcal{O}(d_t^A)$ requests.
    With $\xi=\max\{|E^-_t|, |E^+_t|\}$ it follows $|\sigma_t|\leq \xi\cdot d^A_t$. Hence, the result follows.
\end{proof}
\begin{algorithm}[htb]
    \caption{Streaming algorithm for the STC in temporal networks}
    \label[algorithm]{alg:streaming}
    \Input{Stream of edges arriving in chronological order, $\Delta\in\mathbb{N}$}
    \Output{2-Approx. of \textsc{MinWeightSTC} for each time window of size $\Delta$}
    \BlankLine
    \DontPrintSemicolon
    Initialize $t_s=1$, $t_e=t_s+\Delta-1$\;
    
    Initialize empty list of edges $E_\tau$ and empty aggregated graph $A$\;
    
    \While{temporal edges are incoming}{
        
        Update $E_{\tau}$ for time window $\tau_t=[t_s,t_e]$ such that $\forall e\in E_{\tau}$ it holds $t(e)\in \tau_t$\;
        
        Let $E^-_\tau$ ($E^+_\tau$) be the edges removed from (inserted to) $E_\tau$\;
        
        \If{$E^-_{\tau}\neq \emptyset$ \kwOr $E^+_\tau\neq\emptyset$}{
            
            Use $E^-_{\tau}$ and $E^+_\tau$ to update $A$ and to obtain the update sequence $\sigma_t$\;\label{alg:streaming:computeA}
            
            Call \Cref{alg:pricing} with $\sigma_t$ to obtain the the MWVC approximation $C_t$\;
            
            Output $S_t=\{\{u,v\}\in E(A)\mid n_{u,v}\not\in C_t\}$\;\label{alg:streaming:output}
        }
        
        Move time window forward by increasing $t_s$ and $t_e$\;
    }
\end{algorithm}

\subsubsection{The streaming algorithm for WeightedMinSTC+}\label{sec:stcplusstreaming}
The streaming algorithm presented in the previous section can be straightforwardly adapted for the \textsc{WeightedMinSTC+} problem.
The main difference is that instead of a wedge graph, we need to maintain a wedge hypergraph $H$ as described in~\Cref{sec:stcplus}. 

When the time window moves forward, the updating the aggregated graphs is the same for the \textsc{WeightedMinSTC} and the \textsc{WeightedMinSTC+} problem. After updating the aggregated graph, the wedge hypergraph $H$ for the \textsc{WeightedMinSTC+} problem needs to be updated, which
is done similarly to the update of the wedge graph described in \Cref{sec:streaming}. Instead of normal edges, we insert and remove hyperedges defined by the wedges in the current aggregated graph $A$. For each new wedge $(v,\{u,w\})\in \mathcal{W}(A)$, we add the hyperedge $ \{n_{uv},n_{vw},n'_{uw}\}$ with $n'_{uw}$ being a vertex corresponding to a potential new edge $\{u,w\}\in F\subseteq \textstyle{V \choose 2} \setminus E(A)$ in the aggregated graph. We call $n'_{uw}$ a \emph{new} vertex in $V(H)$.
We add missing (new) vertices and remove isolated (new) vertices as we update the wedge hypergraph.
Furthermore, the vertex weight function of the hypergraph is updated according to the edge weight changes in the aggregated graph.

When inserting or deleting edges from $E(A)$, or when changing edge weights in $A$, we have to additionally consider the weight of the new vertices in the wedge hypergraph $H$.  
To this end, consider an edge $e=\{u,v\}\in E(A)$ that is modified, i.e., inserted, deleted, or its weight is changed due to the moving time window.
Let $W_e$ be the set of wedges in $A$ that contain the edge $e$.
For each wedge $(v,\{u,w\})\in W_e$, there exists a corresponding hyperedge $\{n_{uv},n_{vw},n'_{uw}\}\in E(H)$ with $n'_{uw}$ being a new node whose weight depends on the weight of $e$.
Recall that the weight of $n'_{uw}$ is defined as 
\[
w(n'_{uw})=\alpha \cdot \sum_{(x,\{u,w\})\in \mathcal{W}_f} \frac{w(\{u,x\})+w(\{x,w\})}{|\mathcal{W}_f|}\text{,}
\]
where $\mathcal{W}_f=\{(x,\{u,w\})\mid x\in V(A)\}\subseteq \mathcal{W}(A)$.
When inserting or deleting $e$, the set $\mathcal{W}_f$ needs to be updated accordingly. Furthermore, when updating the weight $w(e)=w(\{u,x\})$ for $x=v$ (including after inserting or deleting $e$), the value of $w(n'_{uw})$ is updated as well.

As our dynamic MWVC algorithm in~\Cref{sec:updatingmwvc} already is stated for $k$-uniform hypergraphs it can be directly applied to update the MWVC in the wedge hypegraph leading to a $3$-approximation together with \Cref{alg:streaming}.

\section{Experiments}\label{sec:experiments}
We compare the weighted and unweighted STC and STC+ on real-world temporal networks and evaluate the efficiency of our streaming algorithm. More specifically, we discuss the following research questions:
\begin{itemize}
    \item[\textbf{Q1.}] How do the weighted and non-weighted versions of the STC and STC+ compare to each other?
    \item[\textbf{Q2.}] What is the impact of the parameter $\alpha$ on the weighted STC+?
    \item[\textbf{Q3.}] How is the efficiency of our streaming algorithm? 
\end{itemize}

\subsection{Algorithms}
We use the following algorithms for computing the weighted STC.
\begin{itemize}
    \item \texttt{ExactW} and \texttt{ExactW+} are the weighted exact computation using the ILPs for the weighted STC and STC+ (see \Cref{appendix:ilp}).
    \item \texttt{Pricing} and \texttt{Pricing+} use the non-dynamic pricing approximation in the wedge graph for the weighted STC and STC+.
    \item \texttt{DynAppr} is our dynamic streaming \Cref{alg:streaming}.
    \item \texttt{STCtime} is a baseline streaming algorithm that recomputes the MWVC with the pricing method for each time window.
\end{itemize}
And, we use the following algorithms for computing the non-weighted STC.
\begin{itemize}
    \item \texttt{ExactNw} and \texttt{ExactNw+} are the exact computations using an ILP  (see~\cite{adriaens2020relaxing,veldt2022correlation}).
    \item \texttt{Matching} is the matching-based approximation of the unweighted vertex cover in the (non-weighted) wedge graph, see~\cite{sintos2014using}.
    \item \texttt{Matching+} is the adapted matching-based approximation of the unweighted vertex cover for  the non-weighted wedge hypergraph.
    \item \texttt{HighDeg} is a $\mathcal{O}(\log n)$ approximation by iteratively adding the highest degree vertex to the vertex cover, and removing all incident edges, see~\cite{sintos2014using}.
\end{itemize}
We implemented all algorithms in C++, using GNU CC Compiler 9.3.0 with the flag \texttt{--O2} and Gurobi 9.5.0 with Python 3 for solving ILPs. All experiments ran on a workstation with an AMD EPYC 7402P 24-Core Processor with 3.35 GHz and 256 GB of RAM running \text{Ubuntu 18.04.3} LTS, and with a time limit of twelve hours. 
Our source code is available at \url{gitlab.com/tgpublic/tgstc}.

\subsection{Data Sets} 
We use the following real-world temporal networks from different domains.
The first three data sets are human contact networks from the \emph{SocioPatterns} project. For these networks, the edges represent human contacts that are recorded using proximity sensors in twenty-second intervals.
The contact networks are available at \url{www.sociopatterns.org/}.
\begin{itemize}
    \item \emph{Malawi} is a contact network of individuals living in a village in rural Malawi~\citep{ozella2021using}. The network spans around 13 days.
    \item \emph{Copresence} is a contact network representing spatial copresence in a workplace over 11 days~\citep{genois2018can}.
    \item \emph{Primary} is a contact network among primary school students over two days~\citep{stehle2011high}.
\end{itemize}
Furthermore, we use four online communication and social networks. 
\begin{itemize}
    \item \emph{Enron} is an email network between employees of a company spanning over 3.6 years~\citep{klimt2004enron}. The data set is available at \url{www.networkrepository.com/}.
    \item \emph{Yahoo} is a communication network available at the \emph{Network Repository}~\citep{nr} (\url{www.networkrepository.com/}). The network spans around 28 days.
    \item \emph{StackOverflow} is based on the stack exchange website StackOverflow~\citep{paranjape2017motifs}. Edges represent answers to comments and questions. The network spans around 7.6 years. 
    It is available at \url{snap.stanford.edu/data/index.html}.
    \item \emph{Reddit} is based on the \emph{Reddit} social network~\citep{hessel2016science,liu2019sampling}.
    A temporal edge $(\{u,v\},t)$ means that a user $u$ commented on a post or comment of user $v$ at time $t$. The network spans over 10.05 years.
    We used a subgraph from the data set provided at \url{www.cs.cornell.edu/~arb/data/temporal-reddit-reply/}.
\end{itemize}

When loading the data sets, we ignore possible self-loops at vertices. 
\Cref{table:datasets_stats} shows the statistics of the data set.
Note that for a wedge graph $W$ of an aggregated graph $A$, $|V(W)|=|E(A)|$, and the number of edges $|E(W)|$ equals the number of wedges in $A$.
For \emph{Reddit} and \emph{StackOverflow} the size of $|E(W)|$ and the number of triangles are estimated using vertex sampling from~\cite{wu2016counting}.

\begin{table}[htb]
    \centering
    \caption{Statistics of the data sets (*estimated). } 
    \label{table:datasets_stats}
    \resizebox{1\linewidth}{!}{ 	\renewcommand{\arraystretch}{1}\setlength{\tabcolsep}{5pt}
        \begin{tabular}{lrrrrrr}
            \toprule
            \multirow{3}{0.5cm}{\vspace*{4pt}\textbf{Data~set}\vspace*{4pt}}&\multicolumn{6}{c}{\textbf{Properties}}\\
            \cmidrule{2-7} 
            \textbf{ }        &    $|V|$ &    $|\tge|$ & $|\mathcal{T}(\tg)|$ & $|V(W)|$  &     $|E(W)|$ & \#Triangles\\ \midrule
            \emph{Malawi}     &     $86$ &    102\,293 &   $43\,438$ &  $347$   &     $2\,254$ &       $441$ \\ 
            \emph{Copresence} &      219 & 1\,283\,194 &     21\,536 & 16\,725  &     549\,449 &    713\,002 \\
            \emph{Primary}    &    $242$ &  $125\,773$ &    $3\,100$ & $8\,317$ &   $337\,504$ &  $103\,760$ \\ 
            \emph{Enron}      &  87\,101 & 1\,147\,126 &    220\,312 & 298\,607 & 45\,595\,540  & 1\,234\,257 \\
            \emph{Yahoo}      & 100\,001 & 3\,179\,718 & 1\,498\,868 & 594\,989 & 18\,136\,435 &    590\,396 \\
            \emph{StackOverflow} & 2\,601\,977	 & 63\,497\,050 & 41\,484\,769 & 28\,183\,518 & *33\,898\,217\,240 & *110\,670\,755 \\
            \emph{Reddit}      & 5\,279\,069 & 116\,029\,037 & 43\,067\,563 & 96\,659\,109 & *86\,758\,743\,921 & *901\,446\,625 \\
            \bottomrule
    \end{tabular}}
\end{table}

\subsection{Comparing Weighted and Non-weighted STC and STC+ (Q1)}

First, we count the number of strong edges and the mean edge weight of strong edges of the first five data sets. \emph{StackOverflow} and \emph{Reddit} are too large for the direct computation. 
We use the contact frequency as the weighting function for the aggregated networks.
\Cref{table:strongpercentage} shows the percentage of strong edges computed using the different algorithms.
The exact computation for \emph{Enron} and \emph{Yahoo} could not be finished within the given time limit.
For the remaining data sets, we observe for the exact solutions that the number of strong edges in the non-weighted case is higher than for
the weighted case. This is expected, as for edge weights of at least one, the number of strong edges in the 
non-weighted STC is an upper bound for the number of strong edges in the weighted STC.
However, when we look at the quality of the STC by considering how the strong edge weights compare to the empirical strength of the connections, we see the benefits of our new approach. 
\begin{table}[htb]
    \caption{Comparison of the weighted and non-weighted STC (OOT---out of time)}
    \begin{subtable}{\linewidth}
        \centering
        \caption{Percentage of strong edges in aggregated graph.} 
        \label{table:strongpercentage}
        \resizebox{0.9\linewidth}{!}{ 	\renewcommand{\arraystretch}{1}\setlength{\tabcolsep}{5mm}
            \begin{tabular}{lccccc}
                \toprule
                \multirow{3}{0.5cm}{\vspace*{4pt}\textbf{Data~set}\vspace*{4pt}}&\multicolumn{2}{c}{\textbf{Weighted}}&\multicolumn{3}{c}{\textbf{Non-weighted}}\\
                \cmidrule(lr){2-3} \cmidrule(lr){4-6}  
                & \texttt{ExactW} & \texttt{Pricing} & \texttt{ExactNw} & \texttt{Matching} & \texttt{HighDeg} \\ \midrule
                \emph{Malawi}     & 30.83 & 29.97 &37.75  & 27.38 & 36.31 \\ 
                \emph{Copresence} &31.12  & 21.37 &37.95  & 29.20 & 35.31 \\ 
                \emph{Primary}    &27.17  & 21.94 &27.83  & 18.99 & 27.35 \\  
                \emph{Enron}      & OOT   & 2.75  & OOT   & 3.28  & 4.61  \\ 
                \emph{Yahoo}      & OOT   & 9.86  & OOT   & 9.98  & 14.29 \\ 
                \bottomrule
            \end{tabular}
        }
    \end{subtable}
    \begin{subtable}{\linewidth}
        \centering
        \caption{Mean edge weights.} 
        \label{table:meanweights}
        \resizebox{1\linewidth}{!}{ 	\renewcommand{\arraystretch}{1}\setlength{\tabcolsep}{2.5mm}
            \begin{tabular}{lrrrrrrrrrr}
                \toprule
                &\multicolumn{4}{c}{\textbf{Weighted}}&\multicolumn{6}{c}{\textbf{Non-weighted}}\\
                \cmidrule(lr){2-5} \cmidrule(lr){6-11} 
                \multirow{3}{0.5cm}{\vspace*{4pt}\textbf{Data~set}\vspace*{4pt}}&\multicolumn{2}{c}{\texttt{ExactW}}&\multicolumn{2}{c}{\texttt{Pricing}}&\multicolumn{2}{c}{\texttt{ExactNw}}&\multicolumn{2}{c}{\texttt{Matching}}&\multicolumn{2}{c}{\texttt{HighDeg}}\\

                \textbf{ }        & Weak  & Strong & Weak   & Strong & Weak   & Strong & Weak   & Strong & Weak   & Strong \\ \midrule
                \emph{Malawi}     & 23.87 & 902.46 &  24.40 & 926.58 & 218.08 & 421.27 & 255.33 & 399.48 & 242.84 & 385.92 \\
                \emph{Copresence} & 20.30 & 78.32  &  46.13 & 189.31 &  27.22 &  56.56 & 58.85  & 120.07 & 57.13  & 112.63 \\
                \emph{Primary}    &  2.73 & 20.48  &   6.58 & 45.50  &   3.34 &  18.49 & 9.32   & 39.88  & 6.19   & 38.84  \\
                \emph{Enron}      &  OOT  &   OOT  &   3.69 & 9.33   &   OOT  &   OOT  & 3.77   & 6.01   & 3.76   & 5.50   \\
                \emph{Yahoo}      &  OOT  &   OOT  &   4.37 & 14.23  &   OOT  &   OOT  & 4.78   & 10.42  & 4.60   & 9.84   \\ 
                
                \bottomrule
            \end{tabular}
        }
    \end{subtable}
\end{table}
An STC labeling with strong edges with high average weights and weak edges with low average weights is favorable. 
The mean weights of the strong and weak edges are shown in \Cref{table:meanweights}.
\texttt{Pricing} leads to the highest mean edge weight for strong edges in almost all data sets.
The mean weight of the strong edges for the exact methods is always significantly higher for \texttt{ExactW} than \texttt{ExactNw}.
The reason is that \texttt{ExactNw} does not consider the edge weights.
Furthermore, it shows the effectiveness of our approach and indicates that the empirical a priori knowledge given by the edge weights is successfully captured by the weighted STC.
To further verify this claim, we evaluated how many of the highest-weight edges are classified as strong.
To this end, we computed the \emph{precision} and \emph{recall} for the top-$100$ weighted edges in the aggregated graph and the set of strong edges.
Let $H$ be the set of edges with the top-$100$ highest degrees. 
The precision is defined as $p=|H\cap S|/|S|$ and the recall as $r=|H\cap S|/|H|$. 
\Cref{fig:precreca} shows the results. Note that the $y$-axis of precision uses a logarithmic scale. The results show that the algorithms for the weighted STC lead to higher precision and recall values for all data sets.

\begin{figure}
    \centering
    \hspace{12mm}\includegraphics[width=.6\textwidth]{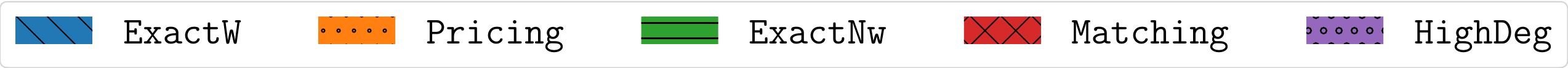}
    \includegraphics[width=.39\linewidth]{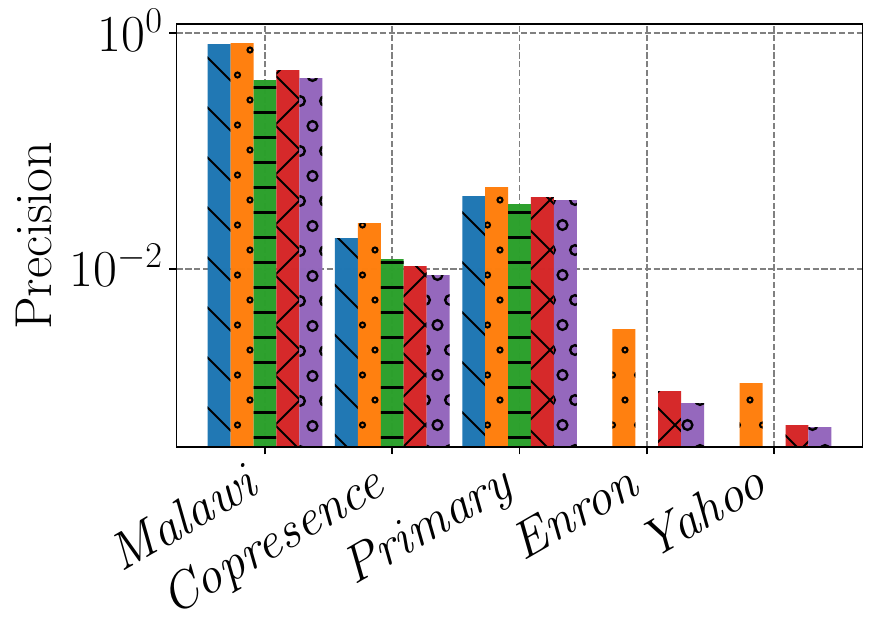}
    \includegraphics[width=.39\linewidth]{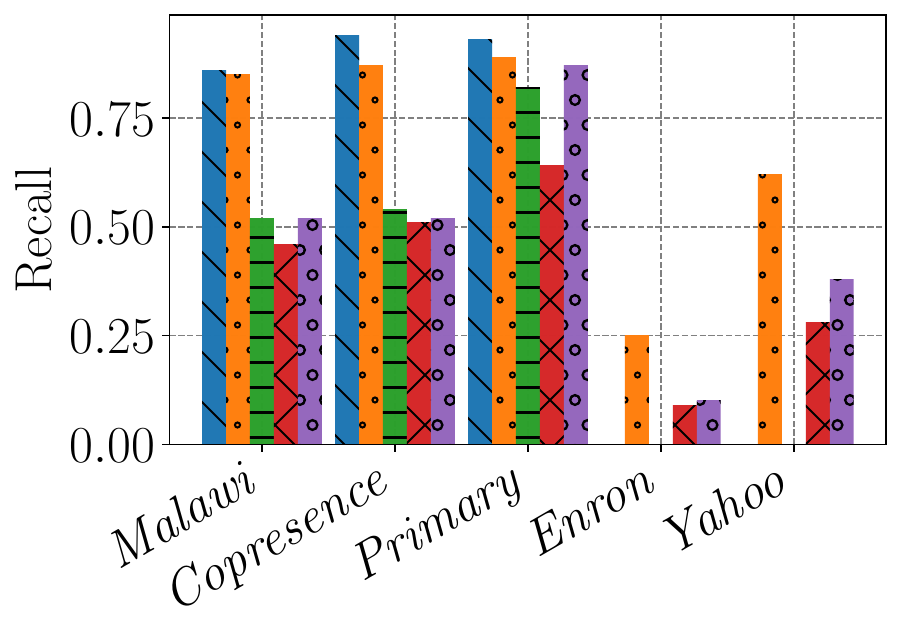}
    \caption{Precision and recall for classifying the top-$100$ highest weighted edges in the aggregated graph as strong edges using the algorithms for the STC. The $y$-axis of precision is logarithmic.}
    \label{fig:precreca}
\end{figure}

Similarly to the comparison of the weighted and unweighted STC, we count the number of strong edges and the mean edge weight of strong edges of the first five data sets using the algorithms for the STC+. 
We set the weighting parameter $\alpha=0.5$ for newly inserted edges in the weighted version.
Due to the increased number of variables, \texttt{ExactNw+} cannot finish the computation for \emph{Primary}.
\Cref{table:strongpercentage_stcplus} shows the percentage of strong edges of the number of edges in the input graph (i.e., not including newly inserted edges) computed using the different algorithms. 
For all data sets and algorithms, there are more strong edges compared to the standard STC variant, where the increase is strongest for \emph{Copresence}. The reason is that by inserting additional weak edges, the number of strong edges can be increased.
The unweighted \texttt{Matching+} approximation leads to slightly more strong edges compared to the weighted \texttt{Pricing+} algorithm because the latter tries to minimize the weight of the weak edges and the former the number of weak edges. %
Moreover, we compare the quality of the STC+ by considering the mean weights of weak and strong edges 
shown in \Cref{table:meanweights_stcplus}.
Compared to STC, the  mean weights of the strong edges can be lower because more strong edges are included in the solution. 
However, for the \emph{Copresence} network, the mean weights of the strong edges are higher for exact solutions.
Similarly to the STC variant, the weighted STC+ approaches, \texttt{ExactW+} and \texttt{Pricing+}, lead to higher quality solutions with higher mean edge weights for the strong edges and lower mean edge weights for the weak edges.

\begin{table}[htb]
    \caption{Comparison of the weighted and non-weighted STC+ (OOT---out of time)}
    \begin{subtable}{\linewidth}
        \centering
        \caption{Percentage of strong edges in aggregated graph.} 
        \label{table:strongpercentage_stcplus}
        \resizebox{0.9\linewidth}{!}{ 	\renewcommand{\arraystretch}{1}\setlength{\tabcolsep}{5mm}
            \begin{tabular}{lcccc}
                \toprule
                \multirow{3}{0.5cm}{\vspace*{4pt}\textbf{Data~set}\vspace*{4pt}}&\multicolumn{2}{c}{\textbf{Weighted}}&\multicolumn{2}{c}{\textbf{Non-weighted}}\\
                \cmidrule(lr){2-3} \cmidrule(lr){4-5}  
                & \texttt{ExactW+} & \texttt{Pricing+} & \texttt{ExactNw+} & \texttt{Matching+} \\ \midrule
                \emph{Malawi}     & 31.70 & 31.12 & 50.72 & 34.29 \\ 
                \emph{Copresence} & 83.04 & 38.00 & 90.73 & 57.27 \\ 
                \emph{Primary}    & 37.39 & 26.46 &   OOT & 32.25 \\
                \emph{Enron}      &   OOT &  3.66 &   OOT &  5.57 \\ 
                \emph{Yahoo}      &   OOT & 12.35 &   OOT & 14.03 \\ 
                \bottomrule
            \end{tabular}
        }
    \end{subtable}
    \begin{subtable}{\linewidth}
        \centering
        \caption{Mean edge weights.} 
        \label{table:meanweights_stcplus}
        \resizebox{1\linewidth}{!}{ 	\renewcommand{\arraystretch}{1}\setlength{\tabcolsep}{2.5mm}
            \begin{tabular}{lrrrrrrrr}
                \toprule
                &\multicolumn{4}{c}{\textbf{Weighted}}&\multicolumn{4}{c}{\textbf{Non-weighted}}\\
                \cmidrule(lr){2-5} \cmidrule(lr){6-9} 
                \multirow{3}{0.5cm}{\vspace*{4pt}\textbf{Data~set}\vspace*{4pt}}&\multicolumn{2}{c}{\texttt{ExactW+}}&\multicolumn{2}{c}{\texttt{Pricing+}}&\multicolumn{2}{c}{\texttt{ExactNw+}}&\multicolumn{2}{c}{\texttt{Matching+}}\\

                \textbf{ }        & Weak  & Strong & Weak   & Strong & Weak   & Strong & Weak   & Strong \\ \midrule
                \emph{Malawi}     & 21.33 & 883.97 & 18.61 & 905.97 & 242.02 & 343.19 & 198.56 & 479.16 \\
                \emph{Copresence} & 27.93 & 86.69 & 31.17 & 151.02 & 38.75 & 80.60 &  46.68 & 99.13 \\
                \emph{Primary}    & 4.83 & 32.36 &  5.35 &  42.24 & OOT & OOT &  8.52 & 28.98 \\
                \emph{Enron}      & OOT & OOT &  3.63 &   9.31 & OOT & OOT &   3.77 & 5.11 \\
                \emph{Yahoo}      & OOT & OOT &  4.15 &  13.68 & OOT & OOT &   4.53 & 10.21 \\ 
                
                \bottomrule
            \end{tabular}
        }
    \end{subtable}
\end{table}

\subsection{Impact of Parameter $\alpha$ on the Weighted STC+ (Q2)}

We computed the exact weighted STC+ using \texttt{ExactW+} for $\alpha\in\{0.001,0.01,0.1,0.5,0.75,0.9\}$ and the \emph{Malawi}, \emph{Copresence}, and \emph{Primary} data sets. 
\Cref{fig:alpha} shows the ratio of strong edges and the normalized mean edge weights of the strong edges.
As expected, for lower values of $\alpha$,  more edges can be classified as strong because more wedges can be closed by adding weak edges. Due to the higher number of strong edges the mean edge weight decreases due to the inclusion of strong edges with low weight. For increasing $\alpha$, the number of strong edges decreases, and the mean edge weight increases.
Therefore, the $\alpha$ parameter can be used to adjust the number of strong edges in a trade-off with the mean edge weight of strong edges.
\begin{figure}
    \centering
    \begin{subfigure}{0.3\linewidth}
        \centering
        \includegraphics[width=1\linewidth]{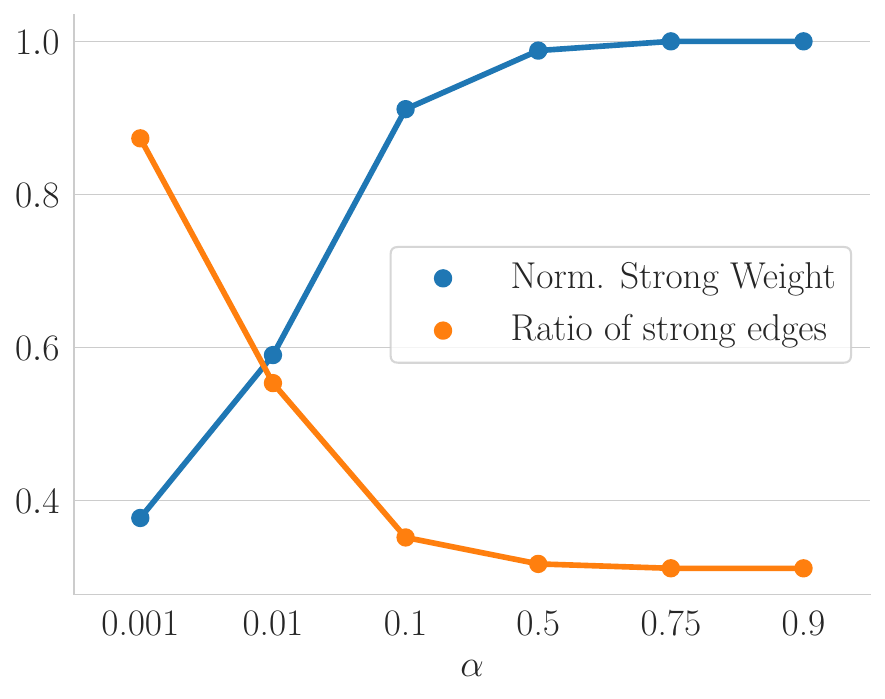}
        \caption{\emph{Malawi}}
        \label{fig:alphaa}
    \end{subfigure}%
    \begin{subfigure}{0.3\linewidth}
        \centering
        \includegraphics[width=1\linewidth]{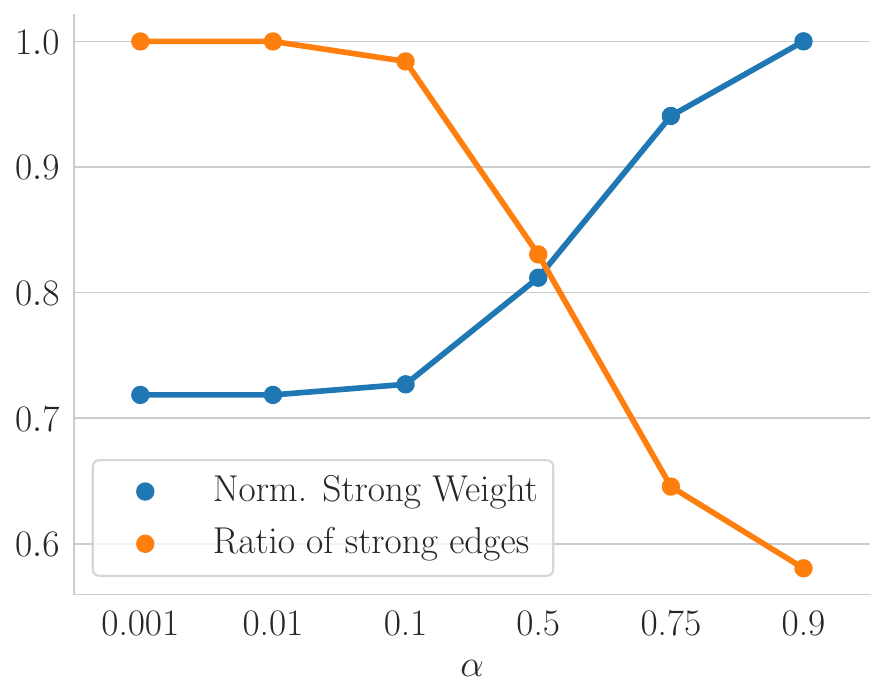}
        \caption{\emph{Copresence}}
        \label{fig:alphab}
    \end{subfigure}%
    \begin{subfigure}{0.3\linewidth}
        \centering
        \includegraphics[width=1\linewidth]{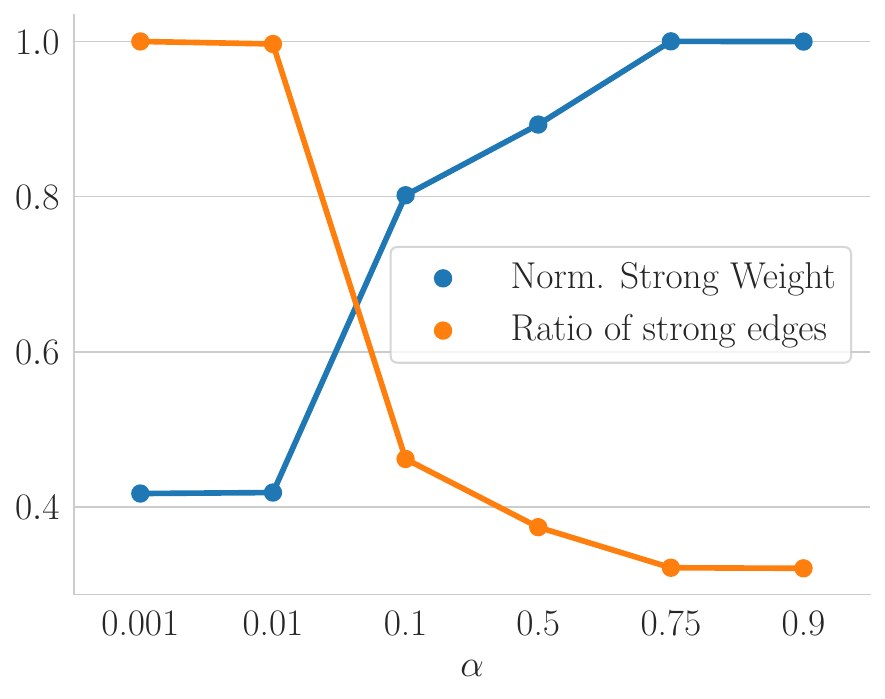}
        \caption{\emph{Primary}}
        \label{fig:alphac}
    \end{subfigure}
    \caption{Effect of $\alpha$ on the ratio of strong edges and the normalizes sum of the weight of strong edges.}
    \label{fig:alpha}
\end{figure}

\subsection{Efficiency of the Streaming Algorithm (Q3)}
In order to evaluate our streaming algorithm, we measured the 
running times on the \emph{Enron, Yahoo, StackOverflow}, and \emph{Reddit} data sets with time window sizes~$\Delta$ of one hour, one day, and one week, respectively.
\Cref{table:stcovertime} shows the results.
In almost all cases, our streaming algorithm \texttt{DynAppr} beats the baseline \texttt{STCtime} with running times that are often orders of magnitudes faster.
The reason is that \texttt{STCtime} uses the non-dynamic pricing approximation, which needs to consider all edges of the current wedge graph in each time window.
Hence, the baseline is often not able to finish the computations within the given time limit, i.e., for seven of the twelve experiments, it runs out of time.
The only case in which the baseline is faster than \texttt{DynAppr} is for the \emph{Enron} data set and a time window size of one hour. Here, the computed wedge graphs of the time windows are, on average, very small (see \Cref{fig:sizesa}), and 
the dynamic algorithm can not make up for its additional complexity due to calling \Cref{alg:pricing}.
However, we also see for \emph{Enron} that for larger time windows, the running times of the baseline strongly increase, and for \texttt{DynAppr}, the increase is slight. 
In general, the number of vertices and edges in the wedge graphs increases with larger time window sizes $\Delta$. Hence, the running times increase for both algorithms with increasing $\Delta$. \Cref{fig:sizes} shows the sizes of the number of vertices $|V(W)|$ and edges $|E(W)|$ in the wedge graphs computed for the time windows of size $\Delta$.
The sizes increase with increasing $\Delta$ because more contacts happen in longer windows.

\begin{figure}[htb]
    \centering
    \begin{subfigure}{0.49\linewidth}
        \includegraphics[width=0.48\linewidth]{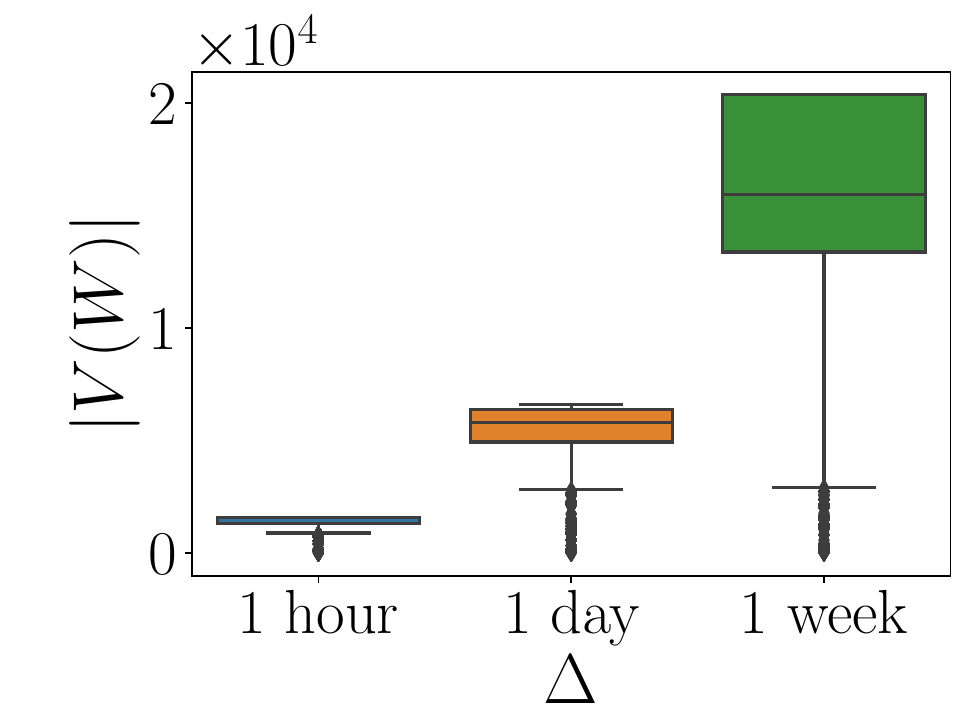}\hfill%
        \includegraphics[width=0.48\linewidth]{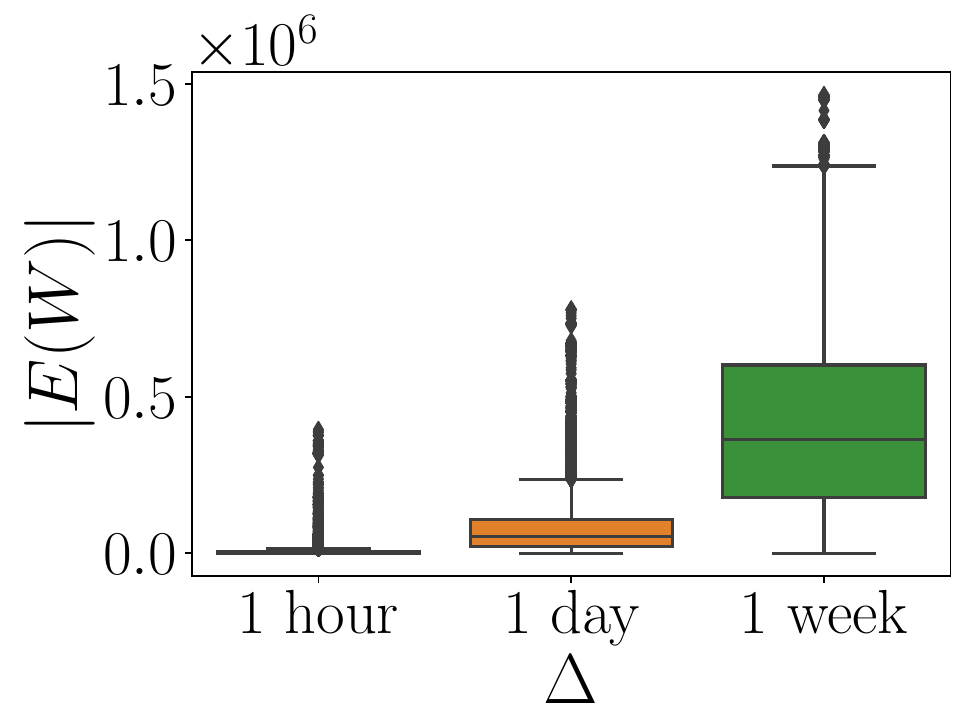}
        \caption{\emph{Enron}}
        \label{fig:sizesa}
    \end{subfigure}\hfill%
    \begin{subfigure}{0.49\linewidth}
        \includegraphics[width=0.48\linewidth]{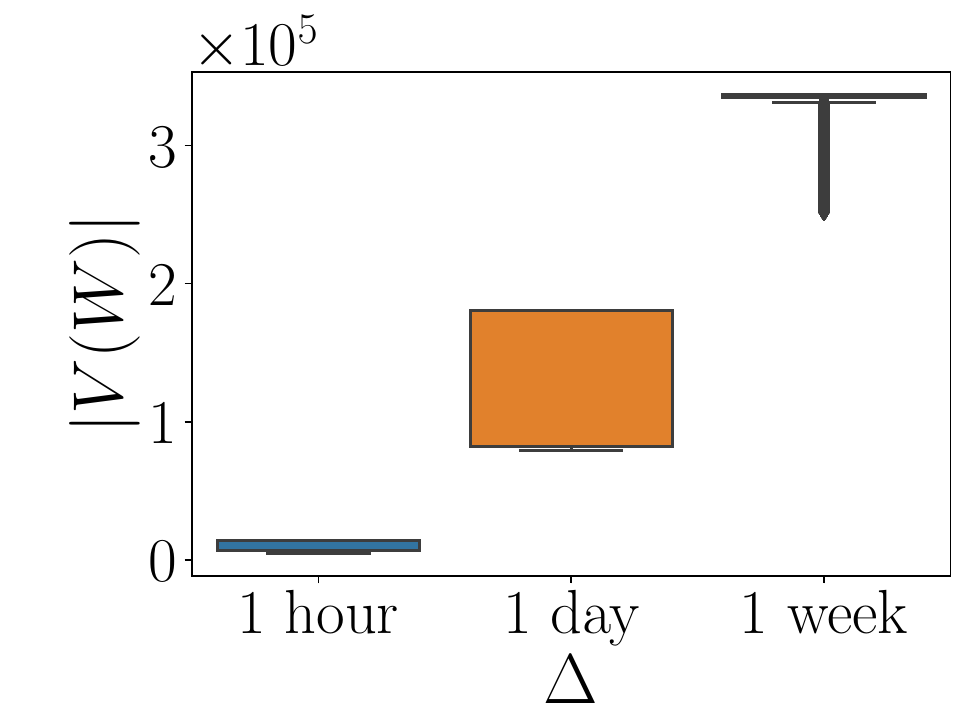}\hfill%
        \includegraphics[width=0.48\linewidth]{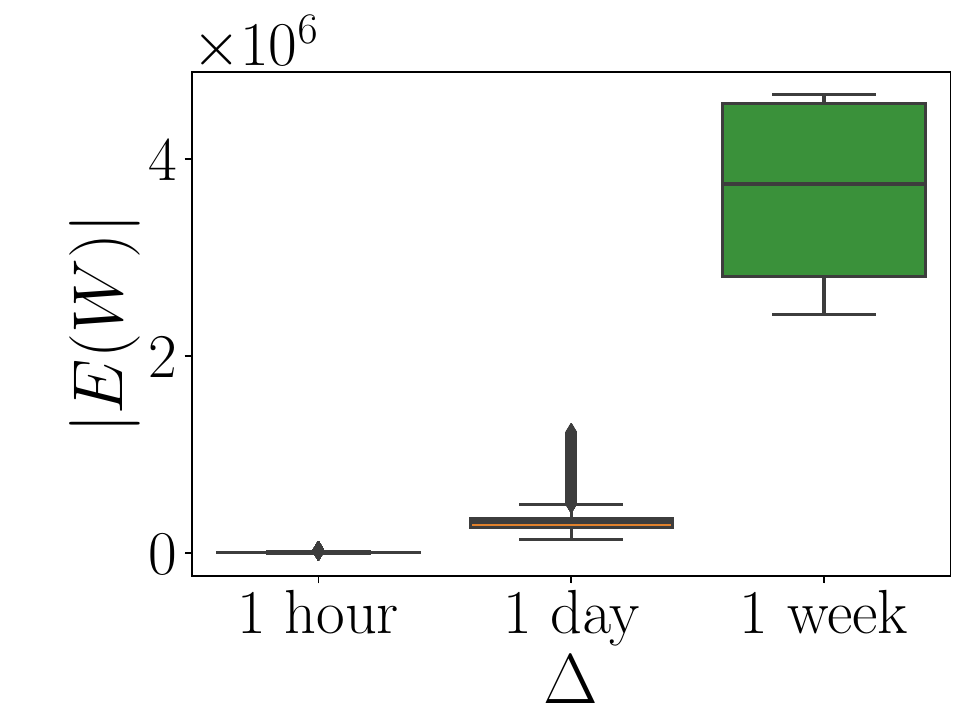}
        \caption{\emph{Yahoo}}
    \end{subfigure}
    \begin{subfigure}{0.49\linewidth}
        \includegraphics[width=0.48\linewidth]{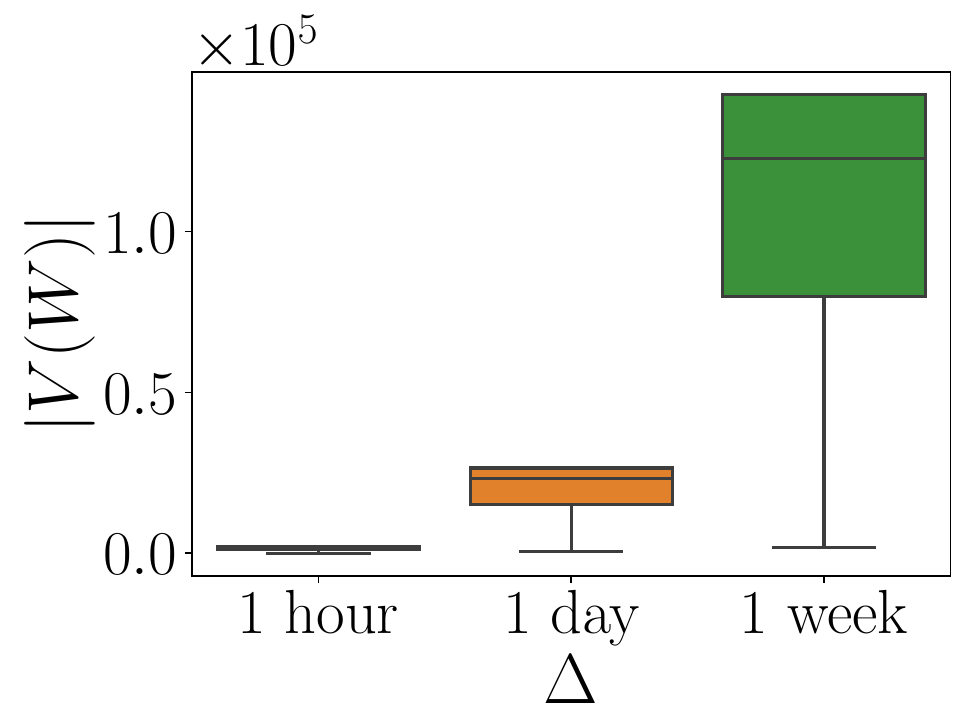}\hfill%
        \includegraphics[width=0.48\linewidth]{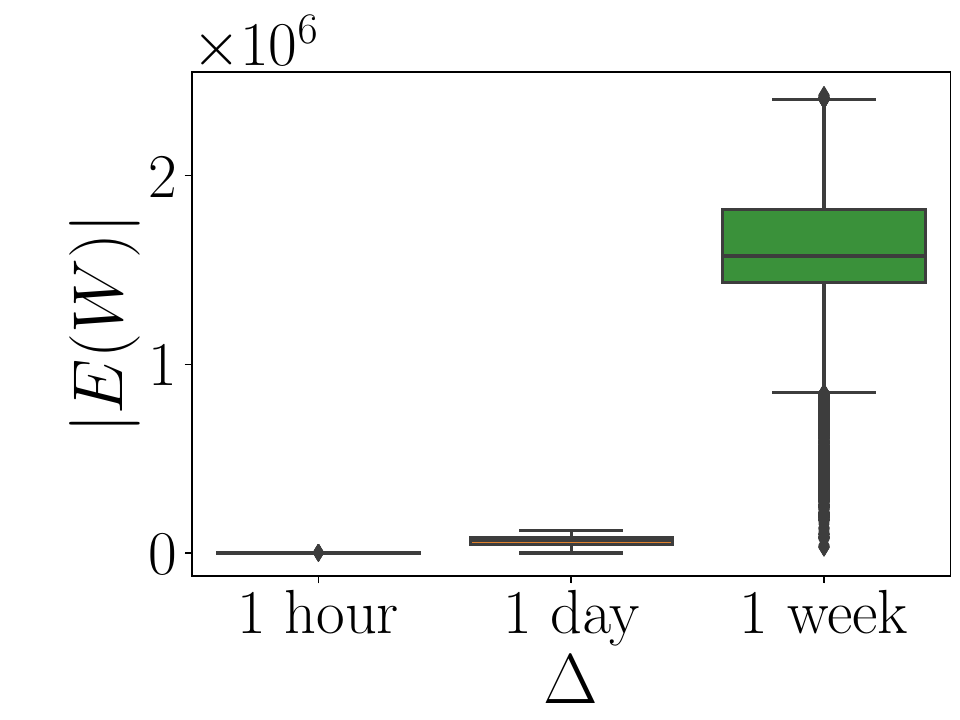}
        \caption{\emph{StackOverflow}}
    \end{subfigure}\hfill%
    \begin{subfigure}{0.49\linewidth}
        \includegraphics[width=0.48\linewidth]{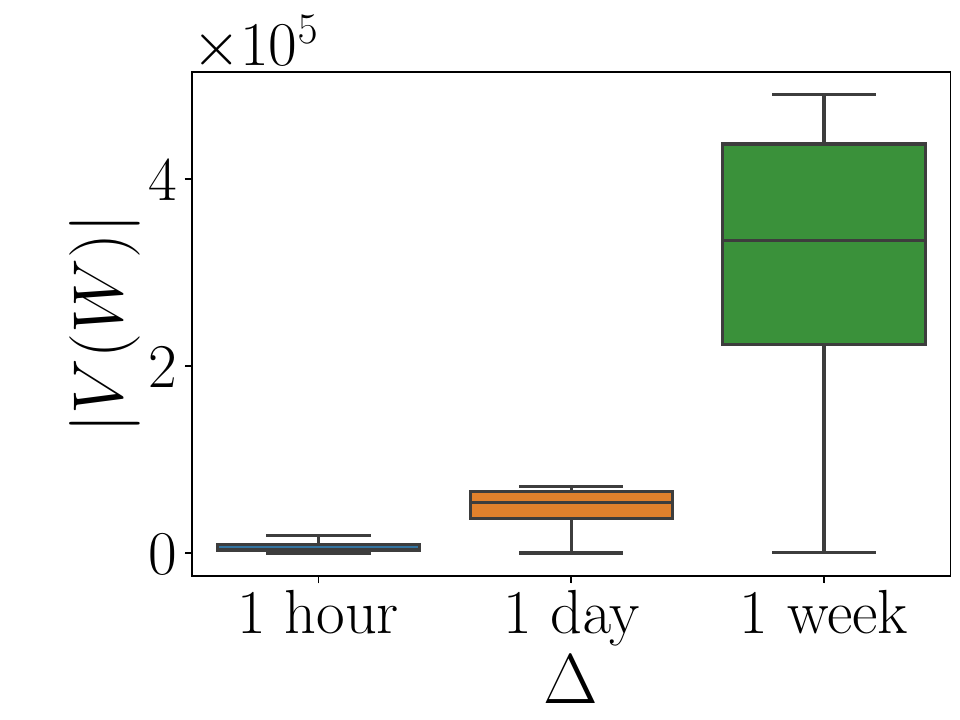}\hfill%
        \includegraphics[width=0.48\linewidth]{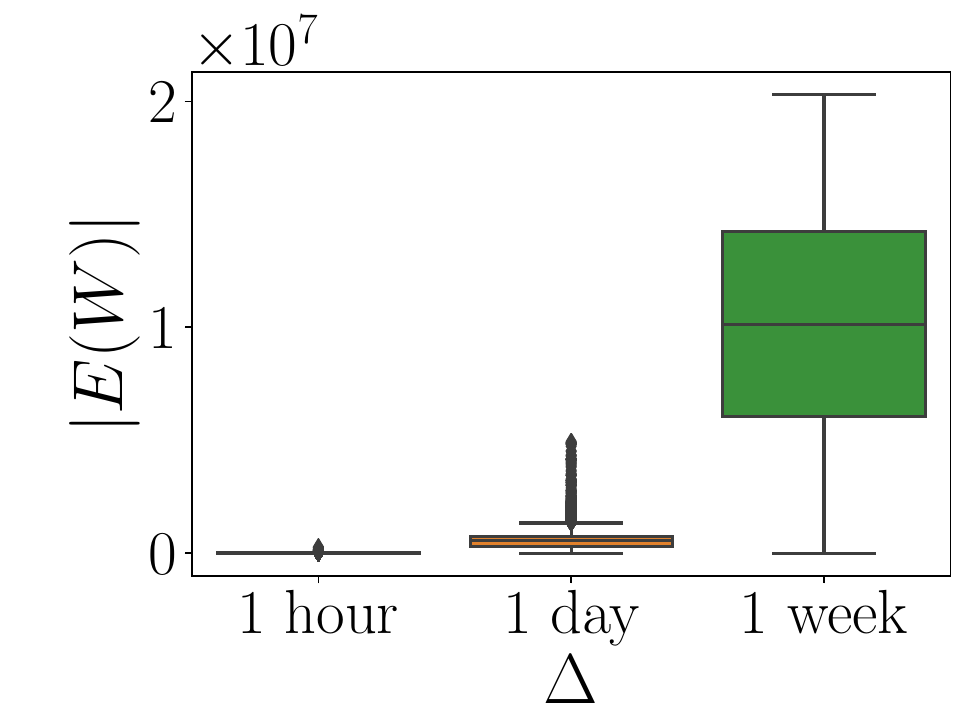}
        \caption{\emph{Reddit}}
    \end{subfigure}
    
    \caption{Boxplots for the sizes of the number of vertices $|V(W)|$ and edges $|E(W)|$ in the wedge graphs computed for the time windows of size $\Delta$.}
    \label{fig:sizes}
\end{figure}
In the case of \emph{Reddit} and a time window size of one week, the sizes of the wedge graphs are too large to compute all solutions within the time limit, even for \textsc{DynAppr}.

\begin{table}[htb]
    \centering
    \caption{Running times in seconds of the streaming alg. (OOT---out of time).} 
    \label{table:stcovertime}
    \resizebox{1\linewidth}{!}{ 	\renewcommand{\arraystretch}{1}\setlength{\tabcolsep}{5mm}
        \begin{tabular}{lrrrrrrrr}
            \toprule
            \multirow{3}{0.5cm}{\vspace*{4pt}\textbf{Data~set}\vspace*{4pt}}&\multicolumn{2}{c}{$\Delta=1$ hour}&\multicolumn{2}{c}{$\Delta=1$ day} & \multicolumn{2}{c}{$\Delta=1$ week}\\
            \cmidrule(lr){2-3} \cmidrule(lr){4-5} \cmidrule(lr){6-7} \cmidrule(lr){8-9}  
            
            \textbf{ }     & \texttt{DynAppr} & \texttt{STCtime} & \texttt{DynAppr} & \texttt{STCtime} & \texttt{DynAppr} & \texttt{STCtime} \\ 
            
            \midrule
            \emph{Enron}         & 264.74 & \textbf{89.18} & \textbf{306.13} & 1\,606.09 & \textbf{352.01} & 20\,870.77 \\
            \emph{Yahoo}         & \textbf{15.99} & 767.40 & \textbf{91.46} & OOT & \textbf{144.52} & OOT \\ 
            \emph{StackOverflow} & \textbf{170.38} & 2\,298.58 & \textbf{971.22} & OOT & \textbf{16\,461.53} & OOT \\
            \emph{Reddit}        & \textbf{1\,254.66} & 13\,244.84 & \textbf{37\,627.79} & OOT & OOT & OOT \\
            \bottomrule
        \end{tabular}
    }
\end{table}

\section{Conclusion and Future Work}\label{sec:conclusion}
We generalized the STC and STC+ to weighted versions to include a priori knowledge in the form of edge weights representing empirical tie strength.
We applied our new STC variants to temporal networks and showed that we obtained meaningful results.
Our main contribution is our 2-approximation (3-approximation) streaming algorithm for the weighted STC (STC+, respectively) in temporal networks. We empirically validated its efficiency in our evaluation.
Furthermore, we introduced a fully dynamic $k$-approximation of the MWVC problem in hypergraphs with $k$-uniform hyperedges that allows efficient updates as part of our streaming algorithm. %

As an extension of this work, a discussion of further variants of the STC or STC+ can be interesting. 
For example,~\cite{sintos2014using} introduced a variant with multiple relationship types. 
Efficient streaming algorithms for weighted versions of this variant are planned as future work.


%

\end{document}